\documentclass[12pt,draftcls,journal,onecolumn]{IEEEtran} 
\usepackage{amsfonts,color,morefloats,pslatex}
\usepackage{amssymb,amsthm, amsmath,latexsym, cases}

\newtheorem{theorem}{Theorem}
\newtheorem{lemma}[theorem]{Lemma}

\newtheorem{corollary}[theorem]{Corollary}

\newtheorem{example}[theorem]{Example}

\newtheorem{conj}[theorem]{Conjecture}

\newcommand{\ord}{{\mathrm{ord}}}

\newcommand{\rank}{{\mathrm{rank}}}
\newcommand{\lcm}{{\mathrm{lcm}}}
\newcommand{\tr}{{\mathrm{Tr}}}
\newcommand{\Tr}{{\mathrm{Tr}}}

\newcommand{\gf}{{\mathrm{GF}}}
\newcommand{\PG}{{\mathrm{PG}}}

\newcommand{\wt}{{\mathtt{wt}}}

\newcommand{\Z}{\mathbb{{Z}}}

\newcommand{\m}{\mathbb{M}}

\newcommand{\C}{{\mathcal{C}}}

\newcommand{\M}{{\mathsf{M}}}

\newcommand{\bc}{{\mathbf{c}}}

\newcommand{\bx}{{\mathbf{x}}}

\newcommand{\bzero}{{\mathbf{0}}}
\newcommand{\bone}{{\mathbf{1}}}

\usepackage{blindtext}

\ifCLASSINFOpdf

\else

\fi

\begin{document}
%
\title{The extended codes of a family of reversible MDS cyclic codes\thanks{Z. Sun's research was supported by The National Natural Science Foundation of China under Grant Number 62002093. C. Ding's research was supported by the Hong Kong Research Grants Council, Proj. No. 16301123.}}

\author{Zhonghua Sun\thanks{Z. Sun with the School of Mathematics, 
Hefei University of Technology, Hefei, 230601, Anhui, China (email:  sunzhonghuas@163.com)}, Cunsheng Ding\thanks{C. Ding is with the Department of Computer Science and Engineering, 
The Hong Kong University of Science and Technology, Clear Water Bay, Kowloon, Hong Kong, China (email: cding@ust.hk)}
}

\maketitle

\begin{abstract} 
A linear code  with parameters $[n, k, n-k+1]$ is called a maximum distance separable (MDS  for short) code. A linear code with parameters $[n, k, n-k]$ is said to be almost maximum distance separable (AMDS for short). A linear code is said to be near maximum distance separable (NMDS for short) if both the code and its dual are AMDS. MDS codes are very important in both theory and practice. There is a classical construction of a $[q+1, 2u-1, q-2u+3]$ MDS code for each $u$ with $1 \leq u \leq \lfloor\frac{q+1}2\rfloor$, which is a reversible and cyclic code. The objective of this paper is to study the extended codes of this family of MDS codes. Two families of MDS codes and several families of NMDS codes are obtained. The NMDS codes have applications in finite geometry, cryptography and distributed and cloud data storage systems. The weight distributions of some of the extended codes are determined. 
\end{abstract}

%
\IEEEpeerreviewmaketitle

\section{Introduction of motivations, objectives and methodology}\label{sec-intro} 

By an $[n, k, d]$ code over $\gf(q)$ we mean a $k$-dimensional linear subspace of $\gf(q)^n$ 
having minimum Hamming distance $d$.  The parameters of a linear code refer to its length, 
dimension and minimum distance. 
An $[n, k, d]$ 
code over $\gf(q)$ is called \textit{distance-optimal} (respectively, 
\textit{dimension-optimal} and \textit{length-optimal}) if there is no $[n, k, d' \geq d+1]$ (respectively, 
$[n, k' \geq k+1, d]$ and $[n' \leq n-1, k, d]$) linear code over $\gf(q)$. An optimal code 
is a code that is length-optimal, or dimension-optimal, or distance-optimal, or meets a bound for linear codes. 
    
Let $A_i(\C)$ denote the number of codewords of Hamming weight $i$ 
in a linear code $\C$. The vector $(A_0(\C), A_1(\C), \ldots, A_n(\C))$ is called the weight distribution of $\C$, and the 
weight enumerator is defined by the polynomial 
$ 
\sum_{i=0}^n A_i(\C)z^i. 
$ 
Throughout this paper, we use $\dim(\C)$ and $d(\C)$ to denote the dimension and minimum distance of a linear code $\C$. 

Denote $[n]=\{0,1,\cdots,n-1\}$ for each positive integer $n$. We usually index the coordinates of the codewords in $\C$ with the elements in $[n]$. 
An $[n, k, d]$ code $\C$ over $\gf(q)$ is called an 
$(n, k, d,q; r)$-LRC (locally recoverable code) if for each $i \in [n]$ there is a subset $R_i \subseteq [n] \setminus \{i\}$  
of size $r$ and a function $f_i(x_1,x_2, \ldots, x_r)$ on $\gf(q)^r$ such that $c_i=f_i(\bc_{R_i})$ for each codeword 
$\bc=(c_0, c_1,\ldots, c_{n-1})$ in $\C$, where $\bc_{R_i}$ is the projection of $\bc$ at $R_i$. The symbol $c_i$ is called 
the $i$-th \emph{code symbol} and the set $R_i$ is called the \emph{repair set} or \emph{recovering set} of the 
code symbol $c_i$. In this definition of LRCs, the degrees of the functions $f_i$ are not restricted. If we require that 
each $f_i$ be a homogeneous function of degree $1$ in the definition above, then we say that $\C$ is $(n, k, d, q; r)$-LLRC 
(linearly local recoverable code) and has linear locality $r$.  
For general theory about minimum locality and minimum linear locality, the reader is referred to \cite{TFZD21}. 

For any $(n, k, d, q; r)$-LLRC, Gopalan et al. proved the following upper bound on the minimum distance $d$ \cite{GHSY12}: 
\begin{equation}\label{eq-bound}
  d\leq n-k-\left\lceil\frac{k}r\right\rceil+2.
\end{equation}
The bound in \eqref{eq-bound} is similar to the Singleton bound, and is thus called the \emph{Singleton-like bound}. If an $(n,k,d,q;r)$-LLRC meets the Singleton-like bound with equality, then we say that the $(n,k,d,q;r)$-LLRC is distance-optimal ($d$-optimal for short). 
If an $(n,k,d,q;r)$-LLRC meets the Singleton-like bound minus one with equality, then we say that the $(n,k,d,q;r)$-LLRC is almost distance-optimal (almost $d$-optimal for short).  
 Note that the Singleton-like bound is not tight for codes over small finite
fields, as it is independent of the alphabet size $q$.

For any $(n, k, d, q; r)$-LLRC, Cadambe and Mazumdar developed the following bound on the dimension $k$ \cite{CM13}, \cite{CM15}:
\begin{equation*}\label{cm-bound}
k\leq\min_{t\in \Z_+}[tr+k_{{\rm opt}}^{(q)}(n-t(r+1),d)],
\end{equation*}
where $\Z_+$ denotes the set of all positive integers, and 
$k_{{\rm opt}}^{(q)}(n, d)$ is the largest possible dimension of a code with length $n$, minimum distance $d$, and alphabet size $q$.   An $(n,k,d,q;r)$-LLRC that attains the
CM bound with equality is said to be dimension-optimal ($k$-optimal for short).

Any linear code with parameters $[n, k, n-k+1]$ for some positive integers $n$ and $k$ is called an MDS (maximum distance separable) code. Any linear code with parameters $[n, k, n-k]$ for some positive integers $n$ and $k$ is said to be almost maximum distance separable (AMDS for short). A linear code is said to be near maximum distance separable (NMDS for short) if both the code and its dual are almost maximum distance separable. By definition, an $[n, k]$ linear code $\C$ is NMDS if and only if $d(\C) + d(\C^\perp)=n$, where $d(\C)$ and $d(\C^\perp)$ denote the minimum distance of $\C$ and $\C^\perp$, respectively.  

MDS codes are very important in theory and have important applications. For example, the Reed-Solomon codes 
are widely used in communication systems and data storage devices. It is easy to verify that each MDS code is a 
$d$-optimal and $k$-optimal LLRC. In theory, MDS codes over finite fields are the same as arcs in projective geometry, 
and some MDS codes are hrperovals and can be used to construct $t$-designs \cite[Chapter 12]{Dingbook18}. 
Motivated by the importance of MDS codes in both theory and practice, our first objective is to construct several families of 
MDS codes by extending a family of known MDS cyclic codes. Our methodology is the extending technique in coding theory.  

The error-correcting capability of NMDS codes is slightly worse than that of MDS codes. However, NMDS codes have the following advantages over MDS codes: 
\begin{itemize}
\item Some infinite families of NMDS codes support infinite families of $t$-designs for $t \in \{2,3,4\}$  \cite{DingTang19,HW22,TangDing20,XCQ22,YZ22}, while MDS codes support only trivial designs. 
\item The length $n$ of an NMDS code over $\gf(q)$ with parameters $[n, k, n-k]$ can be more than $q+1$, while that of an MDS code with parameters $[n, k, n-k+1]$ over $\gf(q)$ is conjectured to be at most $q+1$ for $k \not\in \{3, n-3\}$. 
\end{itemize} 
NMDS codes correspond also to certain objects in projective geometry \cite{AL05,AL08,AGS21,BGZ16,CCMP,DeBoer96,DodLan95,FaldumWillems97,GP07,MMP02}. In addition, NMDS codes have applications in cryptography \cite{MS19}. Recently, it was shown in \cite{TFZD21} that an NMDS code is either a $d$-optimal and $k$-optimal LLRC or an almost $d$-optimal and $k$-optimal LLRC (see \cite{LiHeng22a,LiHeng22b} for further information). This application of NMDS codes in distributed and cloud data storage as LLRCs is the second motivation of this paper. Compared with MDS codes, there are a very small number of families of NMDS codes reported in the literature \cite{AL05, AL08, AGS21, BGZ16, CCMP, DingTang19, DodLan95, HLW22, LiHeng22a, LiHeng22b, TangDing20, WH2020, XCQ22}. Our second objective is to construct several families of NMDS codes by extending a family of known MDS cyclic codes.  

While a lot of infinite families of MDS codes over finite fields with variable parameters have been constructed in the literature, very limited results on the extended codes of MDS codes are known \cite[Chapter 7]{HP03}, \cite[Chapter 11]{MS77}. Our third objective is to study the extended codes of a family of MDS cyclic codes.



The rest of this paper is organised as follows. Section \ref{sec-prelimi} presents some auxiliary results about linear codes. 
Section \ref{sec-LCDMDScycliccode} introduces a family of known MDS codes over $\gf(q)$. Section  \ref{sec-extcodes} 
investigates the extended codes of this family of known MDS codes over $\gf(q)$. Section \ref{sec-final202171} summarizes 
this paper and makes concluding remarks. 
 
\section{Preliminaries}\label{sec-prelimi}
 
\subsection{Some operations on linear codes} 

Let $\C$ be an $[n, k, d]$ code over $\gf(q)$. The \emph{extended code\index{entended code}} 
$\overline{\C}$ of $\C$ is defined by 
$$ 
\overline{\C}=\left\{(c_0,c_1,\ldots, c_{n-1}, c_{n}) \in \gf(q)^{n+1}: (c_0,c_1,\ldots, c_{n-1}) \in \C \mbox{ with } 
\sum_{i=0}^{n} c_i =0\right\}.  
$$

Let $\mathbf{H}$ and $\overline{\mathbf{H}}$ denote the check matrix of $\C$ and $\overline{\C}$, respectively. Then we have the following known result whose proof is straightforward.   

\begin{lemma}\label{lem-extendedCodeParam} \cite{HP03}
Let $\C$ be an $[n, k, d]$ code over $\gf(q)$. Then $\overline{\C}$ is an $[n+1, k, \overline{d}]$ code over $\gf(q)$, where $\overline{d}=d$ or $d+1$. In the binary case, $\overline{d}=d$ if $d$ is even, and $\overline{d}=d+1$ otherwise. In addition, the check matrix $\overline{\mathbf{H}}$ of $\overline{\C}$ can be deduced from that of $\C$ by 
\begin{align*}
\overline{\mathbf{H}}
=\begin{bmatrix} 
\bone & 1 \\
\mathbf{H} &     \bzero^T  
\end{bmatrix},  
\end{align*} 
where $\bone$ is the all-one vector of length $n$ and $\bzero$ is the all-zero vector of length $n$.    
\end{lemma}     

The all-one vector $\bone$ in $\gf(q)^n$ and the rows of a generator matrix of $\C$ span a linear subspace $\widetilde{\C}$ 
of $\gf(q)^n$, which is called the \emph{augmented code} of $\C$. If $\bone  \in \C$, then $\widetilde{\C}=\C$. Otherwise, 
$\C$ is a subcode of $\widetilde{\C}$ and $\dim(\widetilde{\C})=\dim(\C)+1$. The following lemma will be useful. 

\begin{lemma}\label{lem-augment-extend-code} 
Let $q$ be a power of a prime $p$, $n$ be a positive integer with $n\equiv -1\pmod{p}$. Let $\C$ be a linear code of length $n$ over $\gf(q)$. Then $\overline{\widetilde{\C}}=\widetilde{\overline{\C}}$. 
\end{lemma}  

\begin{proof}
Let $\bc=(c_0, c_1, \ldots, c_{n-1}) \in \C$. By definition, $\overline{\C}=\left\{\left(\bc, -\bc \bone^T  \right): \bc \in \C\right\}$, and 
\begin{eqnarray}\label{NEQ1::1}
\widetilde{\overline{\C}}=\left\{\left(a\bone + \bc, a-\bc \bone^T \right): a\in \gf(q), \, \bc \in \C\right\},   
\end{eqnarray} 
where $\bone$ is the all-one vector of length $n$. Similarly, we have 
\begin{eqnarray*}
\widetilde{\C} = \{a\bone+\bc: a \in \gf(q), \, \bc \in \C\}
\end{eqnarray*} 
and 
\begin{eqnarray*}
\overline{\widetilde{\C}} =\left\{\left(a\bone+\bc, -an-\bc \bone^T \right): a \in \gf(q), \, \bc \in \C\right\}.
\end{eqnarray*} 
Note that 
$$ 
\left( a-\bc \bone^T \right)-\left(-an-\bc \bone^T \right)=(n+1)a=0
$$
for each $a \in \gf(q)$. 
The desired result follows. 
\end{proof}

The following lemma will also be needed later. 

\begin{lemma} \label{lem-mixedoperations} 
Let $\C$ be a linear code of length $n$ over $\gf(q)$ and the all-one vector $\bone \in \C$. Then the following hold.
\begin{enumerate}
\item $(\overline{\C})^\perp=\widetilde{\overline{\C^\perp}}$. 
\item If $\bone\notin \C^\perp$ and $\C$ is an MDS code, then $d((\overline{\C})^\perp)=d( \widetilde{\C^\perp})+1$.	
\end{enumerate}
\end{lemma}

\begin{proof}
1) Let $\mathbf{H}$ denote the check matrix of $\C$. Since $\bone \in \C$, the sum of all the coordinates in each codewords of $\C^\perp$ is zero. As a result, we conclude that the code $\overline{\C^\perp}$ has generator matrix $\begin{bmatrix}
\mathbf{H} & \bzero^T	
\end{bmatrix}$, where $\bzero$ is the all-zero vector. The desired conclusion follows from Lemma \ref{lem-extendedCodeParam}.  

2) By Result 1 and Equation (\ref{NEQ1::1}), we deduce that
\begin{eqnarray}\label{NEQ1::2}
(\overline{\C})^\perp=\left\{\left(a\bone + \bc, a \right): a\in \gf(q), \, \bc \in \C^\perp \right\}. 
\end{eqnarray}
Suppose $\C$ is an $[n, k, n-k+1]$ MDS code over $\gf(q)$, then $\C^\perp$ is an $[n, n-k, k+1]$ MDS code over $\gf(q)$. Since $\bone \notin \C^\perp$, we deduce that $\widetilde{\C^\perp}$ is an $[n, n-k+1]$ code over $\gf(q)$. By the Singleton bound, $d(\widetilde{\C^\perp})\leq k=d(\C^\perp)-1$. It follows that there exist $a\in \gf(q)\backslash \{0\}$ and $\bc\in \C^\perp$ such that $\wt(a\bone+\bc)=d(\widetilde{\C^\perp})$. Note that $\left(a\bone + \bc, a \right)\in (\overline{\C})^\perp$, we get that  
\begin{eqnarray}\label{NEQ1::3}
	d((\overline{\C})^\perp )\leq d(\widetilde{\C^\perp})+1.
\end{eqnarray}
On the other hand, let $ \overline{\bc}\in (\overline{\C})^\perp$ and $\overline{\bc}\neq \bzero$. By Equation (\ref{NEQ1::2}), there exist $\bc \in \C^\perp$ and $a\in \gf(q)$ such that $\overline{\bc}=\left(a\bone + \bc, a \right)$. If $a=0$, it follows from $\overline{\bc}\neq \bzero$ that $\bc\neq \bzero$. Consequently, 
\begin{align*}
	\wt(\overline{\bc})&=\wt(\bc)\geq d(\C^\perp) \geq d( \widetilde{\C^{\perp}})+1.
\end{align*}
If $a\neq 0$, $\wt(\overline{\bc})=1+\wt(a\bone +\bc)$. It follows from $\bone \notin \C^\perp$ that $a\bone+\bc\neq \bzero$. Consequently, 
$$\wt(\overline{\bc})\geq d( \widetilde{\C^{\perp}})+1.$$
In summary, 
\begin{eqnarray}\label{NEQ1::4}
	d((\overline{\C})^\perp )\geq d(\widetilde{\C^\perp})+1.
\end{eqnarray}
Combining Inequalities (\ref{NEQ1::3}) and (\ref{NEQ1::4}), we deduce that $d((\overline{\C})^\perp )= d(\widetilde{\C^\perp})+1$. This completes the proof. 
\end{proof}

When $\C$ is a cyclic code of length $n$ over $\gf(q)$, the augmented code of $\C$ has the following result.

\begin{lemma}\label{nlem:1}
Let $\C$ be the cyclic code of length $n$ over $\gf(q)$ with generator polynomial $g(x)$. Then $\widetilde{\C}$ is the cyclic code of length $n$ over $\gf(q)$ with generator polynomial $\gcd(\frac{x^n-1}{x-1}, g(x))$.
\end{lemma}

\begin{proof}
	Let $\C'$ be the cyclic code of length $n$ over $\gf(q)$ with generator polynomial $\frac{x^n-1}{x-1}$, then $$\widetilde{\C}=\C'+\C=\{\bc_1+\bc_2: \bc_1\in \C', \bc_2\in \C \}.$$ The desired result follows from \cite[Theorem 4.3.7]{HP03}. 
\end{proof}

\subsection{Cyclic codes and BCH codes over finite fields} 

An $[n, k, d]$ code $\C$ over $\gf(q)$ is called {\em cyclic} if $(c_0,c_1, \ldots, c_{n-1}) \in \C$ implies 
$$(c_{n-1}, c_0, c_1,\ldots, c_{n-2}) \in \C.$$  
We identify any vector $(c_0,c_1, \ldots, c_{n-1}) \in \gf(q)^n$ with the polynomial 
$$c(x)=\sum_{i=0}^{n-1} c_ix^i \in \gf(q)[x]/(x^n-1).$$  
Then a code $\C$ of length $n$ over $\gf(q)$ corresponds to a subset $\C(x)$ of the ring $\gf(q)[x]/(x^n-1)$, where 
$$ 
\C(x):=\left\{\sum_{i=0}^{n-1} c_ix^i : \bc=(c_0,c_1, \ldots, c_{n-1}) \in \C\right\}. 
$$
A linear code $\C$ is cyclic if and only if $\C(x)$ is an ideal of the ring $\gf(q)[x]/(x^n-1)$. 

Note that every ideal of $\gf(q)[x]/(x^n-1)$ must be principal. Let $\C=( g(x) )$ be a cyclic code of length $n$ over $\gf(q)$, where $g(x)$ is monic and has the smallest degree among all the generators of $\C$. Then $g(x)$ is unique and called the {\em generator polynomial}, and $h(x)=(x^n-1)/g(x)$ is referred to as the {\em check polynomial} of $\C$. 

Let $n$ be a positive integer with $\gcd(n, q)=1$. Let $\Z_n$ denote the ring of integers modulo $n$. For any integer $s$ with $0 \leq s <n$, the \emph{$q$-cyclotomic coset of $s$ modulo $n$\index{$q$-cyclotomic coset modulo $n$}} is defined by 
$$ 
C_s^{(q,n)}=\{s, s q, s q^2, \cdots, sq^{\ell_s-1}\} \bmod n \subseteq \Z_n,  
$$
where $\ell_s$ is the smallest positive integer such that $s \equiv s q^{\ell_s} \pmod{n}$, and is the size of the $q$-cyclotomic coset. The smallest integer in $C_s^{(q, n)}$ is called the \emph{coset leader\index{coset leader}} of $C_s^{(q, n)}$. Let $\Gamma_{(n, q)}$ be the set of all the coset leaders. It is easily seen that $C_s^{(q, n)} \cap C_t^{(q, n)} = \emptyset$ for any two distinct elements $s$ and $t$ in  $\Gamma_{(n, q)}$, and  
\begin{eqnarray*}\label{eqn-cosetPP}
\bigcup_{s \in  \Gamma_{(n,q)} } C_s^{(q, n)} = \Z_n. 
\end{eqnarray*}
Hence, the distinct $q$-cyclotomic cosets modulo $n$ partition $\Z_n$. 

Let $m=\ord_{n}(q)$ be the order of $q$ modulo $n$, and let $\alpha$ be a generator of the group $\gf(q^m)^*$. Put $\beta=\alpha^{(q^m-1)/n}$, then $\beta$ must be a primitive $n$-th root of unity in $\gf(q^m)$. The minimal polynomial $\M_{\beta^s}(x)$ of $\beta^s$ over $\gf(q)$ is the monic polynomial of the smallest degree over $\gf(q)$ with $\beta^s$ as a root and is given by 
\begin{eqnarray*}
\M_{\beta^s}(x)=\prod_{i \in C_s^{(q, n)}} (x-\beta^i) \in \gf(q)[x], 
\end{eqnarray*} 
which is irreducible over $\gf(q)$. 

Let $\delta$ be an integer with $2 \leq \delta \leq n$ and let $h$ be an integer. A \emph{BCH code\index{BCH codes}} over $\gf(q)$ with length $n$ and \emph{designed distance} $\delta$, denoted by $\C_{(q,n,\delta,h)}$, is a cyclic code with generator polynomial 
\begin{eqnarray}\label{eqn-BCHdefiningSet}
\lcm(\M_{\beta^h}(x), \M_{\beta^{h+1}}(x), \cdots, \M_{\beta^{h+\delta-2}}(x)).  
\end{eqnarray}
If $h=1$, the code $\C_{(q,n,\delta,h)}$ with the generator polynomial in (\ref{eqn-BCHdefiningSet}) is called a \emph{narrow-sense\index{narrow sense}} BCH code. If $n=q^m-1$, then $\C_{(q,n,\delta,h)}$ is referred to as a \emph{primitive\index{primitive BCH}} BCH code. 

BCH codes are a subclass of cyclic codes with interesting properties and applications. In many cases BCH codes are the best linear codes. For example, among all binary cyclic codes of odd lengths $n$ with $n \leq 125$ the best cyclic code is always a BCH code except for two special cases \cite{Dingbook15}. Reed-Solomon codes are also BCH codes and are widely used in communication devices and consumer electronics. In the past decade, a lot of progress on the study of BCH codes has been made (see, for example, \cite{LWL19,LiSIAM,LLFLR,SYW,YLLY,ZSH21}).  

A cyclic code $\C$ is said to be reversible if $\C \cap \C^\perp = \{\bzero\}$. Such code is said to be linear complementary dual (LCD) code. The family of MDS codes presented in Section \ref{sec-LCDMDScycliccode} are cyclic BCH codes and reversible.

\subsection{The weight distributions of MDS codes and NMDS codes} 

The weight distribution of MDS codes is known and given by the following lemma \cite[p. 321]{MS77}. We will use it later. 

\begin{lemma}\label{lem-sdjoin1}
Let $\C$ be an $[n, k, d]$ code over $\gf(q)$ with $d=n-k+1$, and let the weight enumerator of $\C$ be 
$1+\sum_{i=d}^{n} A_iz^i$. Then 
$$ 
A_i =\binom{n}{i} (q-1) \sum_{j=0}^{i-d} (-1)^j  \binom{i-1}{j} q^{i-j-d}  
$$
for all $d \leq i \leq n$.
\end{lemma} 

We have the following weight distribution formulas for NMDS codes. 

\begin{lemma}\cite{FaldumWillems97}\label{lem-DLwtd}
Let $\C$ be an $[n, k, n-k]$ NMDS code. Then the weight distributions of $\C$ and $\C^\perp$ are given by 
\begin{eqnarray}\label{eqn-DL281}
A_{n-k+s} = \binom{n}{k-s} \sum_{j=0}^{s-1} (-1)^j \binom{n-k+s}{j}(q^{s-j}-1) + 
             (-1)^s \binom{k}{s}A_{n-k}
\end{eqnarray} 
for $s \in \{1,2, \cdots, k\}$, and 
\begin{eqnarray}\label{eqn-DL282}
A_{k+s}^\perp = \binom{n}{k+s} \sum_{j=0}^{s-1} (-1)^j \binom{k+s}{j}(q^{s-j}-1) + (-1)^s \binom{n-k}{s}A_{k}^\perp 
\end{eqnarray} 
for $s \in \{1,2, \cdots, n-k\}$. 
\end{lemma} 

By definition, we have $\sum_{i=0}^n A_i=q^{k}$ and $\sum_{i=0}^n A_i^\perp=q^{n-k}$. The weight distribution of an NMDS code $\C$ or its dual cannot be determined by Equations (\ref{eqn-DL281}) and (\ref{eqn-DL282}). It was shown in \cite{DingTang19} that two $[n, k, n-k]$ NMDS codes over $\gf(q)$ could have different weight distributions. Thus, the weight distribution of an $[n, k, n-k]$ NMDS code over $\gf(q)$ depends on not only $n$, $k$ and $q$, but also some other parameters of the code. This is a major difference between MDS codes and NMDS codes. Actually, it could be very difficult to determine the weight distribution of an NMDS code. This difficulty will be justified in this paper later.

\section{The family of classical MDS cyclic codes}\label{sec-LCDMDScycliccode} 

Starting from now on, let $q=p^m>2$, where $p$ is a prime and $m$ is a positive integer. Let $\alpha$ be a primitive element of $\gf(q^2)$. Define $\beta=\alpha^{q-1}$, then $\beta$ is a primitive $(q+1)$-th root of unity in $\gf(q^2)$. Let $\m_{\beta^i}(x)$ denote the minimal polynomial of $\beta^i$ over $\gf(q)$. It is easily verified that
$$ x^{q+1}-1=\prod_{i=0}^{\lfloor\frac{q+1}2\rfloor}\m_{\beta^i}(x),$$
where $\m_{\beta^0}(x)=x-1$,
\begin{align*}
\m_{\beta^i}(x)&=(x-\beta^i)(x-\beta^{q+1-i})\\
&=(x-\beta^i)(x-\beta^{-i})	
\end{align*}
for all $1 \leq i \leq \frac{q}2$, and
\begin{align*}
\m_{\beta^{\frac{q+1}2}}(x)&=x-\beta^{\frac{q+1}2}=x+1	
\end{align*}
for $q$ being odd.

For each $u$ with $1 \leq u \leq \lfloor \frac{q+1}2 \rfloor$, define 
\begin{eqnarray*}
g_u(x)=\m_{\beta^u}(x)\m_{\beta^{u+1}}(x) \cdots \m_{\beta^{\lfloor \frac{q+1}2 \rfloor}}(x).  
\end{eqnarray*}
Let $\C_u$ be the cyclic code of length $q+1$ over $\gf(q)$ with generator polynomial $g_u(x)$. We have the following result. 

\begin{theorem}\label{thm-sdjoin1} \cite[Chapter 11]{MS77} \cite{DahlPed1992}
For each $u$ with $1 \leq u \leq \lfloor \frac{q+1}2 \rfloor$, $\C_u$ is a $[q+1, 2u-1, q-2u+3]$ MDS cyclic code and $(\C_u)^\perp$ is a $[q+1, q+2-2u, 2u]$ MDS cyclic code. Both codes are reversible. 
\end{theorem} 

The codes $\C_u$ and $(\C_u)^\perp$ are called the classical MDS cyclic codes.  
Reed-Solomon codes over $\gf(q)$ are MDS codes over $\gf(q)$ with length $q-1$. Hence, the codes $\C_u$ and $(\C_u)^\perp$
in Theorem \ref{thm-sdjoin1} are not Reed-Solomon codes. The  
objective of this paper is to study the extended codes of some of the MDS codes $\C_u$.  

\section{The extended codes of the family of MDS cyclic codes}\label{sec-extcodes} 

The objective of this section is to study the extended codes $\overline{\C_u}$ of the family of MDS codes $\C_u$ introduced in Section \ref{sec-LCDMDScycliccode} and their duals. As will be seen later, we can settle the parameters of the extended code $\overline{\C_u}$ only for a few values of $u$. 

\subsection{Some general results of the extended code $\overline{\C_u}$ and its dual}

By Lemma \ref{lem-extendedCodeParam}, the extended code $\overline{\C_u}$ has length $q+2$, dimension $2u-1$ and 
$$ 
d(\overline{\C_u})= d(\C_u) \mbox{ or } d(\overline{\C_u})=d(\C_u)+1. 
$$ 
As will be seen later, both cases can happen. 

In the case that $d(\overline{\C_u})=d(\C_u)+1$, $\overline{\C_u}$ is an MDS code and the parameters and weight distributions of $\overline{\C_u}$ are easily obtained from Lemma \ref{lem-sdjoin1}. A trivial verification shows that if $u=1$, then $\overline{\C}_u$ is a $[q+2, 1, q+2]$ MDS code over $\gf(q)$. It will be proved later that $\overline{\C_2}$ and $\overline{\C_{q/2}}$ are MDS codes for even $q$. The classical MDS code conjecture leads to the conjecture that $\overline{\C_u}$ is an AMDS code for the other cases.  

In the case that $d(\overline{\C_u})=d(\C_u)$, $\overline{\C_u}$ is a $[q+2, 2u-1, q+3-2u]$ AMDS code over $\gf(q)$, but its weight distribution cannot be derived from Lemma \ref{lem-DLwtd}. As will be seen later, settling the weight distribution of  $\overline{\C_u}$ is very difficult in general. In this case, we have 
\begin{eqnarray*}
d((\overline{\C_u})^\perp) \leq 2u-1. 
\end{eqnarray*} 
If $d((\overline{\C_u})^\perp) =2u-1$, then $\overline{\C_u}$ is an NMDS code. We are much more interested in 
this case. As will be seen later, in some cases we indeed have $d((\overline{\C_u})^\perp) < 2u-1$. The most difficult 
problem in this project is to determine $d((\overline{\C_u})^\perp)$. The following theorem may be useful in determining $d((\overline{\C_u})^\perp)$. 

\begin{theorem}\label{thm-fund21jproj}
Let notation be the same as before. Let $2\leq u\leq \lfloor \frac{q+1}2 \rfloor$. Then 
\begin{eqnarray*}
d( (\overline{\C_u})^\perp ) = d( \C_{(q, q+1, u, 1)} ) +1. 
\end{eqnarray*}
\end{theorem}

\begin{proof}
Recall that $\C_u$ has generator polynomial 
$$ 
g_u(x)=\m_{\beta^u}(x) \m_{\beta^{u+1}}(x) \cdots \m_{\beta^{\lfloor \frac{q+1}2 \rfloor}}(x).  
$$ 
It is easily verified that $(\C_u)^\perp$ has generator polynomial 
$$ 
h_u(x)=(x-1)\m_{\beta}(x)\cdots \m_{\beta^{u-1}}(x).  
$$ 
 It is easily seen that $g_u(x)\mid (\frac{x^{q+1}-1}{x-1})$ and $h_u(x)\nmid (\frac{x^{q+1}-1}{x-1})$, then $\bone \in \C_u$ and $\bone \notin (\C_u)^\perp$. By Result 2 of Lemma \ref{lem-mixedoperations}, we deduce that $d((\overline{\C_u})^\perp )= d (\widetilde{\C_u^\perp})+1$. 
It is easily seen that 
\begin{eqnarray}\label{eqn-tracecudualaug}
\gcd\left( \frac{x^n-1}{x-1}, h_u(x)\right)=\m_{\beta}(x)\m_{\beta^2}(x)\cdots \m_{\beta^{u-1}}(x).	
\end{eqnarray}
It follows from Lemma \ref{nlem:1} that $\widetilde{\C_u^\perp}$ is the narrow-sense BCH code of length $q+1$ over $\gf(q)$ with designed distance $u$.  The desired result follows. 
\end{proof}

It follows from Equation (\ref{eqn-tracecudualaug}) that $\left(\widetilde{\C_u^\perp}\right)^\perp $ is the cyclic code of length $q+1$ over $\gf(q)$ with generator polynomial $(x-1) g_u(x)$, which is subcode of $\C_u$. Specifically, 
\begin{eqnarray}\label{eqn-21j313}
\C_u= \widetilde{\left(\widetilde{\C_u^\perp}\right)^\perp}. 
\end{eqnarray} 
To simplify notation, we will use $\C(u)$ to denote the code $\left(\widetilde{\C_u^\perp}\right)^\perp$. 
With this new notation, Equation (\ref{eqn-21j313}) becomes $\C_u=\widetilde{\C(u)}$. Note that $\C(u)\subseteq \C_u$, we get that $d(\C_u)\leq d(\C(u))$. Hence, $\C(u)$ is a $[q+1, 2u-2, d]$ code over $\gf(q)$, where $d\in \{q-2u+3, q-2u+4 \}$. Furthermore, we have the following result.

\begin{theorem}\label{NTHM29:1}
Let notation be the same as before. Let $2\leq u \leq \lfloor\frac{q+1}2\rfloor$	. Then $d(\overline{\C_u})=d(\C(u))$ and $A_{q-2u+3}(\overline{\C_u})=A_{q-2u+3}(\C(u))$.
\end{theorem}

\begin{proof}Let $\tr(x)$ denote the trace function from $\gf(q^2)$ to $\gf(q)$. By Delsarte's theorem, 
$$\C(u)=\left\{ \bc(a_1, a_2,\ldots, a_{u-1}) : a_i\in \gf(q^2) \right \}, $$
and $\C_u=\left\{ a\bone+\bc(a_1, a_2,\ldots, a_{u-1}) : a\in \gf(q), a_i\in \gf(q^2) \right \}$, where 
$$\bc(a_1, a_2,\ldots, a_{u-1})=\left (\tr\left(\sum_{i=1}^{u-1}a_i  \beta^{i j}  \right)\right)_{j=0}^q.$$
For any $i$ with $1\leq i\leq u-1$, $\beta^i\neq 1$. Then $\sum_{j=0}^q \beta^{i j}=0$. It is easily verified that  
\begin{align*}
	(a\bone +\bc(a_1, a_2,\ldots, a_{u-1})) \bone^T=&~a+\sum_{j=0}^q \tr\left(\sum_{i=1}^{u-1}a_i  \beta^{i j}  \right)\\
	=&~a+ \tr\left(\sum_{i=1}^{u-1}a_i \sum_{j=0}^q \beta^{i j}  \right)\\
	=&~a.
\end{align*}
Therefore, $$\overline{\C_u}=\{ (a\bone +\bc(a_1, a_2,\ldots, a_{u-1}),-a): a\in \gf(q) ,\ a_i\in \gf(q^2) \}.$$
Since $d(\C_u)=q-2u+3$, we deduce that 
$$\wt((a\bone +\bc(a_1, a_2,\ldots, a_{u-1}),-a))=q-2u+3$$
if and only if $a=0$ and $\wt(\bc(a_1, a_2,\ldots, a_{u-1}))=q-2u+3$. As a result, 
\begin{equation}\label{NEQ29:1}
A_{q-2u+3}(\overline{\C_u})=A_{q-2u+3}(\C(u)).	
\end{equation}
Notice that 
$$d(\overline{\C_u}), d(\C(u))\in \{q-2u+3,q-2u+4\},$$
by Equation (\ref{NEQ29:1}), we deduce that $d(\overline{\C_u})=d(\C(u))$. This completes the proof. 
\end{proof}

Theorems \ref{thm-fund21jproj} and \ref{NTHM29:1} show that $\C_u$, $\C_{(q, q+1,u,1)}$, $\overline{\C_u}$ and their duals have the following relationship:
\begin{align*}
\begin{matrix}
\C_u:~[q+1,2u-1,q-2u+3] &\stackrel{{\rm Dual}}\longleftrightarrow &(\C_u)^\perp:~[q+1,q-2u+2, 2u]\\	
\bigcup & & \bigcap\\
\C(u):~[q+1,2u-2,d_2] &\stackrel{{\rm Dual}}\longleftrightarrow &\C_{(q,q+1, u,1)}:~[q+1,q-2u+3,d_1]\\	
d_2\in \{q-2u+3,q-2u+4\} & &u\leq d_1\leq 2u-1\\	
\overline{\C_u}:~[q+2,2u-1,d_2] &\stackrel{{\rm Dual}}\longleftrightarrow &(\overline{\C_u})^\perp:~[q+2,q-2u+3, d_1+1].\\	
\end{matrix}
\end{align*}
Hence, the minimum distance of $(\overline{\C_u})^\perp$ (resp. $\overline{\C_u}$ ) can be determined if the minimum distance of the code $\C_{(q, q+1, u, 1)}$ (resp. $\C(u)$ ) is known and vice versa. The following lemma comes from \cite{TH1970}. 

\begin{lemma}\label{nlem:10}
Let $q$ be a prime power and let $u$ with $2\leq u\leq \lfloor \frac{q+1}2 \rfloor$. Then the following hold.
\begin{enumerate}
\item If $u\mid (q+1)$, then $d(\C_{(q, q+1, u, 1)})=u$.	
\item If $u\nmid (q+1)$, then $d(\C_{(q, q+1, u, 1)})\geq u+1$.	
\end{enumerate}
\end{lemma}

\begin{theorem}\label{nthm:11}
Let $q=p^m>2$, where $p$ is a prime and $m$ is a positive integer. Let $u$ be an integer with $2\leq u\leq \lfloor \frac{q+1}2 \rfloor$. Then the following hold.
\begin{enumerate}
\item If $u\mid (q+1)$, then $\overline{\C_u}$ has parameters $[q+2, 2u-1, q-2u+3]$ and $(\overline{\C_u})^\perp$ has parameters $[q+2, q-2u+3, u+1]$.	
\item If $u\nmid (q+1)$ and $(u+1)\mid (q+1)$, then $\overline{\C_u}$ has parameters $[q+2, 2u-1, d]$, where $d=q-2u+4$ if $u=p=2$ and $m$ is odd, $d=q-2u+3$ otherwise. The dual code $(\overline{\C_u})^\perp$ has parameters $[q+2, q-2u+3, u+2]$.	
\end{enumerate}
\end{theorem}

\begin{proof}
By Theorem \ref{thm-sdjoin1}, $\C_{u}$ has parameters $[q+1, 2u-1, q-2u+3]$. It follows from Lemma \ref{lem-extendedCodeParam} that $\overline{\C_{u}}$ has parameters $[q+2, 2u-1, d]$, where $d \in \{q-2u+3, q-2u+4\}$. By Theorem \ref{thm-fund21jproj}, we obtain that $d((\overline{\C_u})^\perp)=d(\C_{(q, q+1, u, 1)})+1$.

1) Since $u\mid(q+1)$, by Lemma \ref{nlem:10}, we deduce that $d(\C_{(q, q+1, u, 1)} )=u$. Therefore, $(\overline{\C_u})^\perp$ has parameters $[q+2, q+3-2u, u+1]$. Since $u>2$, $(\overline{\C_u})^\perp$ is not MDS. As a result,  $\overline{\C_u}$ is not MDS. Therefore, $d=q-2u+3$.

2) It is easily verified that $\C_{(q, q+1, u+1, 1)}\subseteq \C_{(q, q+1, u, 1)}$. Consequently, 
$$d(\C_{(q, q+1, u, 1)}) \leq d (\C_{(q, q+1, u+1, 1)}).$$  
Since $(u+1)\mid (q+1)$, we deduce that $d (\C_{(q, q+1, u+1, 1)})=u+1$. Since $u\nmid(q+1)$, it follows from Lemma \ref{nlem:10} that $d(\C_{(q, q+1, u, 1)})\geq u+1$. Consequently, $d(\C_{(q, q+1, u, 1)})=u+1$. It then following from Theorem \ref{thm-fund21jproj} that $d((\overline{\C_u})^\perp)=u+2$. In this case, $(\overline{\C_u})^\perp$ is an MDS code if and only if $u=2$. By assumption, $2\nmid(q+1)$ and $3\mid(q+1)$, we deduce that $p=2$ and $m$ is odd. Therefore, we have the following hold.
\begin{itemize}
\item If $u=p=2$ and $m$ is odd, $(\overline{\C_u})^\perp$ is a $[q+2, q-1, 4]$ MDS code and $\overline{\C_u}$ is a $[q+2, 3, q]$ MDS code. 
\item Otherwise, $(\overline{\C_u})^\perp$ is not MDS. Consequently, $\overline{\C_u}$ is not MDS. Hence, $d=q-2u+3$.  
\end{itemize}
This completes the proof. 
\end{proof}

\begin{lemma}\cite{LDM21} \label{nlem:12}
Let $u$ be a power of a prime $p$ and $q=u^e$ with $e\geq 2$. Then $d(\C_{(q, q+1, u, 1)})=u+1$.	
\end{lemma}

\begin{theorem}\label{nthm:13}
Let $q=p^m>2$, where $p$ is a prime and $m$ is a positive integer. If $u= p^s$, where $s$ is a positive integer with $s\mid m$ and $s<m$. Then $\overline{\C_u}$ has parameters $[q+2, 2u-1, d]$, where $d=q-2u+4$ if $u=2$, and $d=q-2u+3$ if $u>2$. The dual code $(\overline{\C_u})^\perp$ has parameters $[q+2, q-2u+3, u+2]$.	
\end{theorem}

\begin{proof}
By Theorem \ref{thm-sdjoin1}, $\C_{u}$ has parameters $[q+1, 2u-1, q-2u+3]$. By Lemma \ref{lem-extendedCodeParam}, $\overline{\C_{u}}$ has parameters $[q+2, 2u-1, d]$, where $d \in \{q-2u+3, q-2u+4\}$. It follows from Theorem \ref{thm-fund21jproj} and Lemma \ref{nlem:12} that $d((\overline{\C_u})^\perp)=u+2$. Hence, $(\overline{\C_u})^\perp$ has parameters $[q+2, q-2u+3, u+2]$. Notice that $(\overline{\C_u})^\perp$ is an MDS code over $\gf(q)$ if and only if $u=2$. Consequently, $d(\overline{\C_u})=q-2u+4$ for $u=2$, and $d(\overline{\C_u})=q-2u+3$ for $u>2$. This completes the proof.
\end{proof}

\begin{example}
Let $q=3^2$. By Theorems \ref{nthm:11} and \ref{nthm:13}, we obtain the following results.
\begin{enumerate}
 \item If $u=2$, the extended code $\overline{\C_u}$ has parameters $[11,3,8]$ and $(\overline{\C_u})^\perp$ has parameters $[11,8,3]$. These two linear codes are optimal codes \cite{Grassl}.	
 \item If $u=3$, the extended code $\overline{\C_u}$ has parameters $[11,5,6]$ and $(\overline{\C_u})^\perp$ has parameters $[11,6,5]$. These two linear codes are optimal codes \cite{Grassl}.	 
 \item If $u=4$, the extended code $\overline{\C_u}$ has parameters $[11,7,4]$ and is an optimal code \cite{Grassl}. The code $(\overline{\C_u})^\perp$ has parameters $[11,4,6]$ and is a distance-almost-optimal code \cite{Grassl}.	
\end{enumerate}	
\end{example}

\subsection{The case $u=2$} 

In this subsection, we study the parameters of $\overline{\C_2}$ and its dual. By Theorems \ref{nthm:11} and \ref{nthm:13}, we have the following results.

\begin{corollary}\label{Cor9::1}
Let $q$ be an odd prime power. Then 	$\overline{\C_2}$ is a $[q+2, 3, q-1]$ NMDS code over $\gf(q)$.
\end{corollary}

\begin{proof}
The desired result follows directly from Result 1 of Theorem \ref{nthm:11}. 	
\end{proof}

\begin{example}
Let $q=5$. Then the extended code $\overline{\C_2}$ has parameters $[7,3,4]$ and $(\overline{\C_2})^\perp$ has parameters $[7,4,3]$. These two linear codes are optimal codes \cite{Grassl}.	
\end{example}


\begin{corollary}\label{thm-sdjoint2}
Let $q=2^m$ with $m\geq 2$. Then $\overline{\C_2}$ is a $[q+2, 3, q]$ MDS code over $\gf(q)$ with weight enumerator 
$$ 
1 + \frac{(q+2)(q^2-1)}{2} z^q + \frac{q(q-1)^2}{2} z^{q+2}. 
$$ 
\end{corollary}
\begin{proof}
The desired conclusion on the parameters of $\overline{\C_2}$ follows from Theorem \ref{nthm:13}. The desired conclusion on the weight distribution follows from Lemma \ref{lem-sdjoin1}. This completes the proof. 	
\end{proof}

Any $[2^m+2, 3, 2^m]$ code over $\gf(2^m)$ is called a \emph{hyperoval code}, as it corresponds to a hyperoval in the Desarguesian projective plane $\PG(2, \gf(2^m))$ \cite[Chapter 12]{Dingbook18}. 
\begin{example} 
Let $q=2^3$. The extended code $\overline{\C_2}$ has parameters $[10, 3, 8]$ and weight enumerator $1+315 z^{8}+196 z^{10}$. The dual code $(\overline{\C_2})^\perp$ has parameters $[10, 7, 4]$. 
\end{example}

\subsection{The case $u=3$} 

In this subsection, we study the parameters of $\overline{\C_3}$ and its dual. Let $\C(3)$ be the dual of the BCH code $\C_{(q, q+1, 3, 1)}$, then $\C(3)$ is the cyclic code of length $q+1$ over $\gf(q)$ with check polynomial $\m_{\beta}(x) \m_{\beta^2}(x)$, where $\beta$ is a primitive $(q+1)$-th root of unity in $\gf(q^2)$. The following lemmas will be needed later.

\begin{lemma}\label{NLEM::17}
Let notation be the same as before. Let $U_{q+1}=\{1, \beta, \cdots, \beta^q \}$. Then $d(\C_{(q, q+1, 3, 1)})\leq 4$ if and only if there are four pairwise distinct elements $x, y, z, w$ in $U_{q+1}$ such that	
\begin{equation*}\label{NNEQ::22}
xy+x z+x w+y z+y w+zw=0.
\end{equation*}
\end{lemma}
\begin{proof}
For any $4$ pairwise distinct elements $x, y, z, w$ in $U_{q+1}$, define a $4\times 4$ matrix $\mathbf{M}$ by  
 \begin{equation*}
 \begin{bmatrix}
   x^{-2}  &  y^{-2} &   z^{-2}   & w^{-2}\\
   x^{-1}  &  y^{-1} &   z^{-1}   & w^{-1}\\
   x       &  y      &   z        & w     \\
   x^2       &  y^2  &   z^2      &  w^2\\
 \end{bmatrix}.	
 \end{equation*}

It is clear that $d(\C_{(q, q+1, 3, 1)})\leq  4$ if and only if there are four pairwise distinct elements $x, y, z, w$ in $U_{q+1}$ such that the system of homogeneous linear equations 
\begin{equation}\label{NNEQ::23}
    \mathbf{M} \bx^T=\bzero
\end{equation}
has a nonzero solution in $\gf(q)^{4}$. By \cite[Lemma 16]{LDM21}, (\ref{NNEQ::23}) has a nonzero solution in $\gf(q)^{4}$ if and only if $\rank(\mathbf{M})< 4$, where $\rank(\mathbf{M})$ denotes the rank of the matrix $\mathbf{M}$. It is easily verified that  
\begin{equation}
\det(\mathbf{M}) 
 = \sigma_{x,y,z,w} \times (xy+xz+xw+yz+yw+zw). \label{NEQ21}
 \end{equation}
 where $\sigma_{x, y, z, w}=\frac{(y-x)(z-x)(z-y)(w-x)(w-y)(w-z)}{(x y z w)^2}$. Notice that $x, y, z, w$ are pairwise distinct, we get that $\sigma_{x, y, z, w}\neq 0$. Therefore, $\rank(\mathbf{M})< 4$ if and only if
  $$x y+x z+x w+y z+y w+z w=0.$$ 
  This completes the proof. 
\end{proof}

The parameters of $\C_{(q, q+1, 3, 1)}$ and $\C(3)$ are documented in the following theorem.


\begin{theorem}\label{NTHM:17}
Let $q\geq 5$ be a prime power. Then the following hold. 
\begin{enumerate}
\item If $3\mid(q+1)$, then $\C_{(q, q+1, 3, 1)}$ has parameters $[q+1, q-3, 3]$ and $\C(3)$ has parameters $[q+1, 4, q-3]$.
\item If $3\nmid(q+1)$, then $\C_{(q, q+1, 3, 1)}$ is a $[q+1, q-3, 4]$ NMDS code over $\gf(q)$.		
\end{enumerate}
\end{theorem}

\begin{proof}
It is clear that $\C_{(q, q+1, 3, 1)}$ has length $q+1$ and dimension $q-3$, and $\C(3)$ has length $q+1$ and dimension $4$, respectively. Notice that $d(\C(3))\in \{q-3, q-2\}$, and $d(\C(3))=q-2$ if and only if $\C(3)$ is MDS. Hence, we only need to prove that $d(\C_{(q, q+1, 3,1)})=3$ for $3\mid(q+1)$ and $d(\C_{(q, q+1, 3,1)})=4$ for $3\nmid(q+1)$. The rest of the proof is divided into the following two cases.

{\it Case 1}: $3\mid(q+1)$. It follows from Lemma \ref{nlem:10} that $d(\C_{(q, q+1, 3, 1)})=3$. The desired result follows.

{\it Case 2}: $3\nmid(q+1)$. Suppose $q=p^m$, where $p$ is a prime and $m$ is a positive integer. If $p\in \{2, 3\}$, the desired result was proved in \cite{DingTang19}. Now suppose $p$ is an odd prime. Since $3\nmid (q+1)$, by Lemma \ref{nlem:10}, $d(\C_{(q, q+1, 3, 1)} )\geq 4$. Below we prove that $d(\C_{(q, q+1, 3, 1)} )\leq  4$. Let $x\in U_{q+1}$, $y=-x$, $z\in U_{q+1}\backslash \{x, -x \}$ and $w=x^2 z^{-1}$. It is easily verified that $x, y, z, w\in U_{q+1}$ and 
$$xy+x z+x w+y z+y w+zw=0.$$ Since $q$ is odd, we deduce that $y\neq x$. We now prove that $x^2 z^{-1} \in U_{q+1}\backslash \{x, -x, z \}$. 
\begin{itemize}
\item If $x^2 z^{-1}=x$, we get that $z=x$, a contradiction.
\item If $x^2 z^{-1}=y$, we get that $z=-x$, a contradiction.	
\item If $x^2 z^{-1}=z$, we deduce that $z=x$ or $-x$, a contradiction.
\end{itemize}
By Lemma \ref{NLEM::17}, $d(\C_{(q, q+1, 3, 1)})\leq 4$. In summary, $d(\C_{(q, q+1, 3, 1)})=4$. This completes the proof.
\end{proof}

The parameters of $\overline{\C_3}$ are documented in the following theorem.

\begin{theorem}\label{NTHM18}
Let $q\geq 5$ be a prime power. Then the following hold.
\begin{enumerate}
\item If $3\mid(q+1)$, then 	$\overline{\C_3}$ has parameters $[q+2, 5, q-3]$ and $(\overline{\C_3})^\perp$ has parameters $[q+2, q-3, 4]$.
\item If $3\nmid(q+1)$, then $\overline{\C_3}$ is a $[q+2, 5, q-3]$ NMDS code over $\gf(q)$. 
\end{enumerate}
\end{theorem}

\begin{proof}
By Theorem \ref{thm-sdjoin1}, $\C_{3}$ has parameters $[q+1, 5, q-3]$. Then $\overline{\C_{3}}$ has length $q+2$ and dimension $5$. Consequently, $(\overline{\C_3})^\perp$ has length $q+2$ and dimension $q-3$. By Theorems \ref{NTHM29:1} and \ref{NTHM:17}, we get that $d(\overline{\C_3})=d(\C(3))=q-3$. By Theorems \ref{thm-fund21jproj} and \ref{NTHM:17}, we obtain that
\begin{align*}
	d((\overline{\C_3})^\perp)=d(\C_{(q, q+1, 3, 1)})+1=\begin{cases}
	4~&{\rm if}~3\mid(q+1),\\
	5~&{\rm if}~3\nmid(q+1). 
\end{cases}
\end{align*}
This completes the proof.  	
\end{proof}

\begin{example}
Let $q=5^2$. Then the following hold.
\begin{enumerate}
\item The code $\C_3$ has parameters $[26, 5, 22]$ and $(\C_3)^\perp$ has parameters $[26,21, 6]$.	
\item The BCH code $\C_{(q, q+1,3,1)}$ has parameters $[26, 22, 4]$ and $\C(3)$ has parameters $[26, 4, 22]$.
\item The extended code $\overline{\C_3}$ has parameters $[27, 5, 22]$ and $(\overline{\C_3})^\perp$ has parameters $[27, 22, 5]$. 
\end{enumerate}
\end{example}

By Theorem \ref{NTHM29:1}, we have 
\begin{equation}\label{NEQ:23}
A_{q-3}(\overline{\C_3})	=A_{q-3}(\C(3)).
\end{equation}
When $q>5$ and $q\not\equiv 2\pmod{3}$, by Theorem \ref{NTHM18}, $\overline{\C_3}$ is an NMDS code over $\gf(q)$. From Lemma \ref{lem-DLwtd} and Equation (\ref{NEQ:23}), the weight distribution of $\overline{\C_3}$ is determined by $A_{q-3}(\C(3))$. It seems difficult to determine $A_{q-3}(\C(3))$ in general. We will determine the weight distribution of $\overline{\C_3}$ for two special cases. The following two lemmas will be needed later. 

\begin{lemma}\label{lem-sdjoint3} \cite{DingTang19}
Let $q=2^m$ with $m \geq 4$ being even. Then $\C(3)$ is a $[q+1, 4, q-3]$ NMDS code with the following weight distribution: 
\begin{itemize}
\item $A_0(\C(3))=1$. 
\item $A_{q-3}(\C(3))=  \frac{(q-4)(q-1)q(q+1)}{24}$. 
\item $A_{q-2}(\C(3))=  \frac{(q-1)q(q+1)}{2}$. 
\item $A_{q-1}(\C(3))= \frac{(q+1)q^2(q-1)}{4} $. 
\item $A_q(\C(3))= \frac{(q-1)(q+1)(2q^2+q+6)}{6} $. 
\item $A_{q+1}(\C(3))=  \frac{3q^4 - 4q^3 - 3q^2 + 4q}{8}$. 
\item $A_i(\C(3))=0$ for other $i$. 
\end{itemize}
\end{lemma} 

\begin{lemma}\label{nlem20} \cite{DingTang19}
Let $q=3^m$ with $m \geq 2$. Then $\C(3)$ is a $[q+1, 4, q-3]$ NMDS code with the following weight distribution:
\begin{itemize}
\item $A_0(\C(3))=1$. 
\item $A_{q-3}(\C(3))=  \frac{(q-1)^2 q(q+1)}{24}$. 
\item $A_{q-1}(\C(3))=  \frac{(q-1)q(q+1)(q+3)}{4}$. 
\item $A_{q}(\C(3))= \frac{(q^2-1)(q^2-q+3)}{3} $. 
\item $A_{q+1}(\C(3))=  \frac{3 (q-1)^2q (q+1)}{8}$. 
\item $A_i(\C(3))=0$ for other $i$. 
\end{itemize}
\end{lemma} 

The following family of NMDS codes were claimed but not proved in \cite{TFZD21}. Here we prove they 
are indeed NMDS and settle the weight distribution of $\overline{\C_3}$.   

\begin{theorem}\label{thm-sdjoint3}
Let $q=2^m$ with $m \geq 4$ being even. Then $\overline{\C_3}$ is a $[q+2, 5, q-3]$ NMDS code with the following weight distribution: 
\begin{itemize}
\item $A_{0}(\overline{\C_3})= 1$. 
\item $A_{q-3}(\overline{\C_3})=  \frac{q^4 - 4q^3 - q^2 + 4q}{24}$. 
\item $A_{q-2}(\overline{\C_3})=  \frac{q^5 - 4q^4 + 17q^3 + 4q^2 - 18q}{24}$. 
\item $A_{q-1}(\overline{\C_3})=  \frac{3q^4 - 4q^3 - 3q^2 + 4q}{4}$. 
\item $A_{q}(\overline{\C_3})=  \frac{3q^5 - 2q^4 + 17q^3 + 14q^2 - 20q - 12}{12}$. 
\item $A_{q+1}(\overline{\C_3})=  \frac{8q^5 + 13q^4 - 44q^3 - 13q^2 + 36q}{24}$. 
\item $A_{q+2}(\overline{\C_3})=  \frac{3q^5 - 8q^4 + 7q^3 - 2q}{8}$. 
\item $A_{i}(\overline{\C_3})= 0$ for other $i$. 
\end{itemize}
\end{theorem}

\begin{proof}
The conclusions on the parameters of $\overline{\C_3}$ and $(\overline{\C_3})^\perp$ follow from Theorem \ref{NTHM18}. It remains to prove the weight distribution of the code $\overline{\C_3}$. It follows from Lemma \ref{lem-sdjoint3} and Equation (\ref{NEQ:23}) that
\begin{equation}\label{NEQ:24}
A_{q-3}(\overline{\C_3})	=\frac{(q-4)(q-1)q(q+1)}{24}.
\end{equation}
Recall that $\overline{\C_3}$ is an NMDS code over $\gf(q)$. From Lemma \ref{lem-DLwtd} and Equation (\ref{NEQ:24}), the desired weight distribution follows. 
\end{proof}

\begin{example} 
Let $q=2^4$. Then $\overline{\C_3}$ has parameters $[18, 5, 13]$ and weight enumerator 
\begin{eqnarray*}
1+2040z^{13}+  35700z^{14}+  44880z^{15}+ 257295z^{16}+ 377400z^{17}+ 
331260z^{18}. 
\end{eqnarray*} 
The dual code $(\overline{\C_3})^\perp$ has parameters $[18, 13, 5]$. Hence, $\overline{\C_3}$ is NMDS in this case. 
\end{example} 

 Similar to the proof of Theorem \ref{thm-sdjoint3}, one can prove the following. 

\begin{theorem}\label{NTHM23}
Let $q=3^m$ with $m \geq 2$. Then $\overline{\C_3}$ is a $[q+2, 5, q-3]$ NMDS code with the following weight distribution: 
\begin{itemize}
\item $A_{0}(\overline{\C_3})= 1$. 
\item $A_{q-3}(\overline{\C_3})=  \frac{q^4-q^3-q^2+q}{24}$. 
\item $A_{q-2}(\overline{\C_3})=  \frac{q^5 - 4q^4 + 2q^3 + 4q^2 - 3q}{24}$. 
\item $A_{q-1}(\overline{\C_3})=  \frac{3q^4 +q^3 - 3q^2 -q}{4}$. 
\item $A_{q}(\overline{\C_3})=  \frac{3q^5 - 2q^4 + 2q^3 + 14q^2 - 5q - 12}{12}$. 
\item $A_{q+1}(\overline{\C_3})=  \frac{8q^5 + 13q^4 - 29q^3 - 13q^2 + 21q}{24}$. 
\item $A_{q+2}(\overline{\C_3})=  \frac{3q^5 - 8q^4 + 6q^3 - q}{8}$. 
\item $A_{i}(\overline{\C_3})= 0$ for other $i$. 
\end{itemize}
\end{theorem}

\begin{example} 
Let $q=3^2$. Then $\overline{\C_3}$ has parameters $[11, 5, 6]$ and weight enumerator 
\begin{eqnarray*}
1+240z^{6}+  1440z^{7}+  5040z^{8}+ 13880z^{9}+ 22320z^{10}+ 16128z^{11}. 
\end{eqnarray*} 
The dual code $(\overline{\C_3})^\perp$ has parameters $[11, 6, 5]$. Hence, $\overline{\C_3}$ is NMDS in this case. 
\end{example} 

\subsection{The case $u=4$} 

In this subsection, we study the parameters of $\overline{\C_4}$ and its dual. Before studying $\overline{\C_4}$, we need to do some preparations. Let $\C(4)$ denote the dual of the BCH code $\C_{(q, q+1, 4,1)}$, then $\C(4)$ is the cyclic code of length $q+1$ over $\gf(q)$ with check polynomial $\m_{\beta}(x) \m_{\beta^2}(x) \m_{\beta^3}(x)$, where $\beta$ is a primitive $(q+1)$-th root of unity in $\gf(q^2)$. When $q=2^m$ with $m\geq 4$, the parameters of the BCH code $\C_{(q, q+1, 4,1)}$ were studied in \cite{TangDing20}.

\begin{lemma}\label{lem-sdjoint7} \cite{TangDing20}
Let $q=2^m$ with $m \geq 5$ being odd. Then $\C(4)$ is a $[q+1, 6, q-5]$ NMDS code over $\gf(q)$ and  
 $$ 
 A_{q-5}(\C(4))=\frac{(q-1)(q-8)}{30} \binom{q+1}{4}. 
 $$  
\end{lemma}  

\begin{theorem}\label{thm-sdjoint4}
Let $q=2^m$ with $m \geq 4$. Then the following hold. 
\begin{enumerate}
\item If $m$ is even, then $\overline{\C_4}$ has parameters $[q+2, 7, q-5]$ and $(\overline{\C_4})^\perp$ has parameters $[q+2, q-5,6]$. 
\item If $m$ is odd, then $\overline{\C_4}$ is a $[q+2, 7, q-5]$ NMDS code with the following weight distribution: 
\begin{itemize}
\item $A_{0}(\overline{\C_4})= 1$. 
\item $A_{q-5}(\overline{\C_4})= \frac{q^6 - 11q^5 + 25q^4 - 5q^3 - 26q^2 + 16q }{720} $. 
\item $A_{q-4}(\overline{\C_4})=  \frac{q^7 - 11q^6 + 75q^5 - 155q^4 + 24q^3 + 166q^2 - 100q}{720}$. 
\item $A_{q-3}(\overline{\C_4})=  \frac{3q^6 - 17q^5 + 27q^4 + q^3 - 30q^2 + 16q}{48}$. 
\item $A_{q-2}(\overline{\C_4})=  \frac{3q^7 - 7q^6 + 101q^5 - 133q^4 - 64q^3 + 140q^2 - 40q}{144}$. 
\item $A_{q-1}(\overline{\C_4})=  \frac{8q^7 + 47q^6 - 69q^5 + 167q^4 + 141q^3 - 214q^2 - 80q}{144}$. 
\item $A_{q}(\overline{\C_4})=  \frac{45q^7 + 93q^6 + 127q^5 - 255q^4 - 200q^3 + 402q^2 + 28q - 240}{240}$. 
\item $A_{q+1}(\overline{\C_4})=  \frac{264q^7 + 151q^6 - 845q^5 + 415q^4 + 181q^3 - 566q^2 + 400q}{720}$. 
\item $A_{q+2}(\overline{\C_4})=  \frac{53q^7 - 133q^6 + 99q^5 - 19q^4 + 8q^2 - 8q}{144}$. 
\item $A_{i}(\overline{\C_4})= 0$ for other $i$. 
\end{itemize} 	
\end{enumerate}
\end{theorem}

\begin{proof}
1) The first desired result follows directly from Theorem \ref{nthm:13}.

2) With Lemma \ref{lem-DLwtd}, Theorem \ref{NTHM29:1}  and Lemma \ref{lem-sdjoint7}, similar to the proof of Theorem \ref{thm-sdjoint3}, we can prove the second desired result. Details are omitted here.  	
\end{proof}

\begin{example} 
Let $q=2^4$.  Then $\overline{\C_4}$ has parameters $[18, 7, 11]$ and is AMDS. The dual code $(\overline{\C_4})^\perp$ has parameters $[18, 11, 6]$, but is not AMDS. 
\end{example} 

\begin{example} 
Let $q=2^5$. Then $\overline{\C_4}$ has parameters $[34, 7, 27]$ and weight enumerator 
\begin{eqnarray*}
1+ 1014816z^{27}+  34588312z^{28}+  55814880z^{29}+ 686184752z^{30}+ \\ 
2244500192z^{31}+  6875142087z^{32}+ 12784990240z^{33}+ 11677503088z^{34}. 
\end{eqnarray*} 
The dual code $(\overline{\C_4})^\perp$ has parameters $[34, 27, 7]$. Hence, $\overline{\C_4}$ 
is NMDS in this case. 
\end{example} 

Theorem \ref{thm-sdjoint4} documents another family of NMDS codes in this paper. 



Starting from now on, let $q \geq 9$ be an odd prime power. Let $\tr(x)$ denote the trace function from $\gf(q^2)$ to $\gf(q)$. To determine the parameters of $\C_{(q, q+1, 4, 1)}$ and $\C(4)$, we need the following preparations.

Suppose $x_1,x_2,x_3,x_4,x_5$ are five pairwise distinct elements in $U_{q+1}=\{1, \beta, \cdots,\beta^q \}$. Define a $6\times 5$ matrix $\mathbf{M}$ by 
\begin{align*}
\begin{bmatrix}
x_1^{-3} & x_2^{-3} & x_3^{-3} & x_4^{-3} & x_5^{-3}\\
x_1^{-2} & x_2^{-2} & x_3^{-2} & x_4^{-2} & x_5^{-2}\\
x_1^{-1} & x_2^{-1} & x_3^{-1} & x_4^{-1} & x_5^{-1}\\
x_1      & x_2      & x_3      & x_4      & x_5     \\	
x_1^{2}  & x_2^{2}  & x_3^{2}  & x_4^{2}  & x_5^{2}\\
x_1^{3}  & x_2^{3}  & x_3^{3}  & x_4^{3}  & x_5^{3}\\
\end{bmatrix}.
\end{align*}
For any $i\in \{\pm 1, \pm 2,\pm 3\}$, let $\mathbf{M}[i]$ denote the submatrix of $\mathbf{M}$ obtained by deleting the row $(x_1^i, x_2^i, x_3^i, x_4^i, x_5^i)$ of the matrix $\mathbf{M}$. Let $I=\{1,2,3,4,5\}$, the {\it elementary symmetric polynomial} (ESP) of degree $i$ in $5$ variables $x_1,x_2,\cdots, x_5$, written $\sigma_{5,i}$, is defined by 
$$\sigma_{5,i}=\sum_{J\subseteq I, |J|=i} \prod_{j\in J} x_j. $$

\begin{lemma}\label{NLEM2::33}
Let $q\geq 9$ be an odd prime power and $4\mid(q-1)$. Then $d(\C_{(q, q+1, 4,1)})=5$ if and only if there are five pairwise distinct elements $x_1, x_2, x_3, x_4,x_5$ in $U_{q+1}$ such that	$\sigma_{5,2}=0$.
\end{lemma}

\begin{proof}
Since $4\mid(q-1)$, we get that $4\nmid{(q+1)}$. By Lemma \ref{nlem:10}, we deduce that $d(\C_{(q, q+1, 4, 1)})\geq  5$. Therefore, $d(\C_{(q, q+1, 3, 1)})=5$ if and only if there are five pairwise distinct elements $x_1, x_2, x_3, x_4,x_5$ in $U_{q+1}$ such that the system of homogeneous linear equations 
\begin{equation}\label{NEQ1::17}
    \mathbf{M} \bx^T=\bzero
\end{equation}
has a nonzero solution in $\gf(q)^{5}$. By \cite[Lemma 16]{LDM21}, (\ref{NEQ1::17}) has a nonzero solution in $\gf(q)^{5}$ if and only if $\rank(\mathbf{M})\leq 4$. 

It is easily verified that
\begin{align*}
\det(\mathbf{M}[-3])&=\frac{\prod_{1\leq j<i\leq 5}(x_i-x_j)}{\sigma_{5,5}} \times \sigma_{5,2}^q,\\
\det(\mathbf{M}[-2])&=\frac{\prod_{1\leq j<i\leq 5}(x_i-x_j)}{\sigma_{5,5}} \times (\sigma_{5,1}^q \sigma_{5,2}^q-\sigma_{5,5}^q \sigma_{5,2})	,\\
\det(\mathbf{M}[-1])&=\frac{\prod_{1\leq j<i\leq 5}(x_i-x_j)}{\sigma_{5,5}} \times (\sigma_{5,2}^{2q}-\sigma_{5,5}^q \sigma_{5,1}^q \sigma_{5,2})	,\\
\det(\mathbf{M}[1])&=\frac{\prod_{1\leq j<i\leq 5}(x_i-x_j)}{\sigma_{5,5}^3} \times (\sigma_{5,2}^2-\sigma_{5,5} \sigma_{5,1} \sigma_{5,2}^q),\\
\det(\mathbf{M}[2])&=\frac{\prod_{1\leq j<i\leq 5}(x_i-x_j)}{\sigma_{5,5}^3} \times (\sigma_{5,1} \sigma_{5,2}-\sigma_{5,5} \sigma_{5,2}^q),\\
\det(\mathbf{M}[3])&=\frac{\prod_{1\leq j<i\leq 5}(x_i-x_j)}{\sigma_{5,5}^3} \times \sigma_{5,2}.
\end{align*}
It then follows that $\rank(\mathbf{M})\leq 4$ if and only if $\sigma_{5,2}=0$. This completes the proof. 
\end{proof}

\begin{lemma}\label{nlem30::33}
Let $q\geq 9$ be an odd prime power and $4\mid(q-1)$. Then the following hold.
\begin{enumerate}
\item $\tr(\beta^i)\neq \tr(\beta^j)$ for any $i\neq j$ with $1\leq i, j\leq \frac{q-1}2$.
\item $\tr(\beta^i)\notin  \{0,2,-2\}$ for any $1\leq i\leq \frac{q-1}2$.
\item $\tr(\beta^i)=-1$ with $1\leq i\leq \frac{q-1}2$ if and only if $3\mid(q+1)$ and $i=\frac{q+1}3$.
\end{enumerate}
\end{lemma}

\begin{proof}
1) It is easily seen that 
\begin{equation*}
x^{q+1}-1=(x-1)(x+1)\prod_{i=1}^{\frac{q-1}2} (x^2-\tr(\beta^i)x+1).
\end{equation*}
Since $\gcd(q+1, q)=1$, we deduce that $x^{q+1}-1$ has no multiple roots. Hence, $\tr(\beta^i)\neq \tr(\beta^j)$ for any $i\neq j$ with $1\leq i,j\leq \frac{q-1}2$. 

2) Since $4\mid(q-1)$, we deduce that 
$$x^{q+1}-1=x^2(x^{q-1}-1)+x^2-1\equiv x^2-1\pmod{x^2+1}.$$
It follows that $\gcd(x^{q+1}-1, x^2+1)=\gcd(x^2-1, x^2+1)=1$. It then follows that $\tr(\beta^i)\neq 0$ for any $1\leq i\leq \frac{q-1}2$. We now prove that $\tr(\beta^i)\notin \{2, -2\}$ for any $1\leq i\leq \frac{q-1}2$. Since 
$$\tr(\beta^{\frac{q+1}2-i})=\tr(-\beta^{-i})=-\tr(\beta^i),$$ to prove the desired result we only need to prove $\tr(\beta^i)\neq 2$ for any $1\leq i\leq \frac{q-1}2$. It is easily verified that $\tr(\beta^i)=2$ with $1\leq i\leq \frac{q-1}2$ if and only if $\beta^{2i}-2\beta^i+1=0$, i.e., $\beta^i=1$ with $1\leq i\leq \frac{q-1}2$, a contradiction. The desired result follows. 

3) It is clear that $\tr(\beta^i)=-1$ if and only if $\beta^{2i}+\beta^{i}+1=0$. Notice that 
$$\beta^{3i}-1=(\beta^{i}-1)(\beta^{2i}+\beta^i+1),$$
and $\beta^i\neq 1$ for any $1\leq i\leq \frac{q-1}2$. Hence, $\beta^{2i}+\beta^{i}+1=0$ with $1\leq i\leq \frac{q-1}2$ if and only if $\beta^{3i}=1$. It follows that $\tr(\beta^i)=-1$ with $1\leq i\leq \frac{q-1}2$ if and only if $3\mid(q+1)$ and $i=\frac{q+1}{3}$. This completes the proof.	
\end{proof}

\begin{lemma}\label{NLEM::34}
Let $q=p^m$ and $4\mid(q-1)$, where $p$ is an odd prime and $m$ is a positive integer. Then the following hold.
\begin{enumerate}
\item For any $j$ with $1\leq j\leq q$, $$\sum_{i=1}^{\frac{q-1}2} \tr(\beta^{ji})=-1-(-1)^j.$$	
\item If $3\mid(q-1)$, 
$$-\sum_{i=1}^{\frac{q-1}2} (\tr(\beta^i)+1)^{-1}=\frac{q-1}{6}\bmod{p},$$
where $i \bmod{p}$ denotes the unique integer $j$ such that $0\leq j\leq p-1$ and $p\mid(i-j)$.  
\item If $3\mid (q+1)$, 
 $$-\sum_{i=1, \atop i\neq \frac{q+1}3}^{\frac{q-1}2} (\tr(\beta^i)+1)^{-1}=\frac{q-1}{2}\bmod{p}.$$ 
\end{enumerate}	
\end{lemma}

\begin{proof}
1) It is easily seen that 
\begin{align*}
\sum_{i=1}^{\frac{q-1}2} \tr(\beta^{ji})&=\sum_{i=1}^{\frac{q-1}2}[ (\beta^j)^i+(\beta^{j})^{q+1-i}]\\
&=\sum_{i=0}^{q} (\beta^j)^i-1-(\beta^j)^{\frac{q+1}2}\\
&=-1-(-1)^j
\end{align*}
for any $1\leq j\leq q$. The desired result follows.  

2) If $3\mid(q-1)$, then 
$$x \frac{x^{q+1}-1}{x-1}=(x^2+x+1)\sum_{j=0}^{\frac{q-1}3}x^{3j} -1.$$
Hence, 
$$(\beta^{2i}+\beta^i+1) \sum_{j=0}^{\frac{q-1}3}\beta^{3ij}=1$$
for each $1\leq i\leq q$. It follows that 
\begin{align*}
(\tr(\beta^i)+1)^{-1}&=\beta^i(\beta^{2i}+\beta^i+1)^{-1}\\
&=\sum_{j=0}^{\frac{q-1}3} \beta^{(3j+1)i}\\
&= \sum_{j=0}^{\frac{q-7}{6}}(\beta^{(3j+1)i} +\beta^{-(3j+1)i})+\beta^{(\frac{q+1}2)i}\\
&=(-1)^i+ \sum_{j=0}^{\frac{q-7}{6}} \tr(\beta^{(3j+1)i}), 
\end{align*}
where the third equation follows from the fact that $\beta^{[3 (\frac{q-1}3-j )+1]i}=\beta^{(q-3j)i}=\beta^{-(3j+1)i}$. Consequently, 
\begin{align*}
-\sum_{i=1}^{\frac{q-1}2}(\tr(\beta^i)+1)^{-1}&=-\sum_{i=1}^{\frac{q-1}2}[(-1)^i+\sum_{j=0}^{\frac{q-7}6}\tr(\beta^{(3j+1)i})]\\
&=-\sum_{i=1}^{\frac{q-1}2}\sum_{j=0}^{\frac{q-7}6}\tr(\beta^{(3j+1)i})\\
&=-\sum_{j=0}^{\frac{q-7}6} \sum_{i=1}^{\frac{q-1}2} \tr(\beta^{(3j+1)i})\\
&=\sum_{j=0}^{\frac{q-7}6}[1+(-1)^{3j+1}]\\
&=\frac{q-1}{6} \bmod{p},
\end{align*}
where the fourth equation follows from Result 1. The desired result follows.

3) If $3\mid(q+1)$, then
$$(x^2+x+1) \sum_{j=0}^{\frac{q-5}3}(2+3j) x^{3j}(x-1)+3\sum_{j=0}^{\frac{q-2}3} x^{3j}=1.$$

Let $i$ be an integer with $1\leq i\leq \frac{q-1}2$ and $i\neq \frac{q+1}3$, then $\sum_{j=0}^{\frac{q-2}3} \beta^{3ij}=0$. It follows that
 $$(\beta^{2i}+\beta^i+1) \sum_{j=0}^{\frac{q-5}3}(2+3j) \beta^{3ij}(\beta^i-1)=1.$$
Consequently, 
\begin{align*}
(\tr(\beta^i)+1)^{-1}&=\beta^i(\beta^{2i}+\beta^i+1)^{-1}\\
&=\sum_{j=0}^{\frac{q-5}3}(2+3j)(\beta^{(2+3j)i}-\beta^{(1+3j)i})\\
&=\sum_{j=0}^{\frac{q-5}3}(2+3j)\beta^{(2+3j)i}-\sum_{j=1}^{\frac{q-2}3}[2+3(\frac{q-2}{3}-j)]\beta^{[1+3(\frac{q-2}{3}-j)]i}\\
&=\sum_{j=0}^{\frac{q-5}3}(2+3j)\beta^{(2+3j)i}+\sum_{j=1}^{\frac{q-2}3}3j \beta^{-(2+3j)i}\\
&=2\sum_{j=0}^{\frac{q-5}3}\beta^{(2+3j)i}+3\sum_{j=1}^{\frac{q-5}3}j  \tr(\beta^{(2+3j)i})-2\beta^i\\
&=-2\tr(\beta^i)+3\sum_{j=1}^{\frac{q-5}3}j  \tr(\beta^{(2+3j)i}),
\end{align*}
where the sixth equation follows from the fact that 
\begin{align*}
\sum_{j=0}^{\frac{q-5}3}\beta^{(2+3j)i}-\beta^i&=\beta^{2i}\frac{ (\beta^{3i})^{\frac{q-2}3}-1}{\beta^{3i}-1}-\beta^i\\
&=\frac{\beta^{qi}-\beta^{2i}}{\beta^{3i}-1}-\beta^i\\
&=-\tr(\beta^i).
\end{align*}
Therefore, 
\begin{align*}
-\sum_{i=1, \atop i\neq \frac{q+1}3}^{\frac{q-1}2}(\tr(\beta^i)+1)^{-1}&=\sum_{i=1, \atop i\neq \frac{q+1}3}^{\frac{q-1}2}[2\tr(\beta^i)-3\sum_{j=1}^{\frac{q-5}3}j  \tr(\beta^{(2+3j)i})]\\
&=\sum_{i=1}^{\frac{q-1}2}[2\tr(\beta^i)-3\sum_{j=1}^{\frac{q-5}3}j  \tr(\beta^{(2+3j)i})]-2\tr(\beta^{\frac{q+1}3})+3\sum_{j=1}^{\frac{q-5}3}j \tr(\beta^{\frac{2(q+1)}3} )\\
&=3\sum_{j=1}^{\frac{q-5}3}j[1+(-1)^{2+3j}]+2-3\sum_{j=1}^{\frac{q-5}3}j\\
&=\frac{q-1}2 \bmod{p}.
\end{align*}
This completes the proof. 
\end{proof}

\begin{lemma}\label{NLEM::35}
Let $q\geq 9$ be an odd prime power, $4\mid(q-1)$ and $3\nmid(q+1)$. Then there exist two distinct integers $i$ and $j$ with $1\leq i, j\leq \frac{q-1}2$ such that $(\tr(\beta^i)+1)(\tr(\beta^j)+1)+1=0$.
\end{lemma}

\begin{proof}
If $5\mid(q+1)$, then 
\begin{align*}
(\tr(\beta^{\frac{q+1}5})+1)(\tr(\beta^{\frac{2(q+1)}{5}})+1)+1=2\sum_{i=0}^4 \beta^{\frac{i(q+1)}{5}}=0.
\end{align*}
The desired result follows. Below we suppose $5 \nmid(q+1)$. 

Let 
$$\Delta=\left \{i: \ 1\leq i\leq \frac{q-1}2 \right\},$$ 
then $|\Delta|=\frac{q-1}2$. By Result 3 of Lemma \ref{nlem30::33}, $\tr(\beta^i)\neq -1$ for any $i\in \Delta$. Let 
$$S_1=\left\{\tr(\beta^i)+1: \ i\in \Delta \right\}$$ 
and $S_2=\left\{ -(\tr(\beta^i)+1)^{-1}: \ i\in \Delta \right\}$. By Result 1 of Lemma \ref{nlem30::33}, we deduce that $|S_1|=|S_2|=|\Delta|$. It is clear that 
\begin{equation}\label{NEQ3::18}
|S_1\cap S_2|=|S_1|+|S_2|-|S_1\cup S_2|=2|\Delta|-|S_1\cup S_2|.	
\end{equation}
It follows from Result 2 of Lemma \ref{nlem30::33} that $S_1\cup S_2\subseteq \gf(q)\backslash \{-1, 0, 1\}$. It follows from Equation (\ref{NEQ3::18}) that $|S_1\cap S_2|\geq q-1-(q-3)=2$. Define 
$$N:=\left|\left\{i\in \Delta: (\tr(\beta^i)+1)^2+1=0 \right\}\right|,$$
 then $N\in \{0,1,2\}$. Clearly, to prove the desired conclusion, it suffices to prove $|S_1\cap S_2|\geq 3$ or $N\leq 1$.  
 
 Suppose on the contrary that $|S_1\cap S_2|=2$ and $N=2$, then $S_1 \cup S_2=\gf(q)\backslash \{-1,0,1\}$ and $S_1 \cap S_2=\{a, -a\}$, where $a\in \gf(q)^*$ and $a^2=-1$. We seek a contradiction by considering two cases.

{\it Case 1}: $p=3$. It is easily verified that 
\begin{align*}
(\tr(\beta^i)+1)^2+1&=\beta^{2i}+\beta^{-2i}-\beta^i-\beta^{-i}+1	\\
&=\beta^{-2i}(\beta^{4i}-\beta^{3i}+\beta^{2i}-\beta^{i}+1) \\
&=\beta^{-2i}\left(\frac{\beta^{5i}+1}{\beta^i+1}\right).
\end{align*}
It follows that $(\tr(\beta^i)+1)^2+1=0$ with $i\in \Delta$ if and only if $10\mid(q+1)$ and $i\in \{\frac{q+1}{10},\frac{3(q+1)}{10}\}$. By assumption, $5\nmid(q+1)$. Hence, $N=0$, i.e., we get a contradiction.  

{\it Case 2}: $p>3$. Since $3\nmid(q+1)$, we deduce that $3\mid(q-1)$. Since $S_1\cup S_2=\gf(q)\backslash \{-1,0,1\}$ and $S_1 \cap S_2 =\{ a, -a \}$, we have 
\begin{align}
0=\sum_{s\in S_1\cup S_2} s&=\sum_{s\in S_1} s+\sum_{s\in S_1} s-\sum_{s\in S_1\cap S_2} s \notag \\
&=\sum_{i=1}^{\frac{q-1}2}(\tr(\beta^i)+1)-\sum_{i=1}^{\frac{q-1}2}(\tr(\beta^i)+1)^{-1} \notag \\
&=(\frac{q-1}2+ \frac{q-1}{6}) \bmod{p}, \label{NEQ6::18}
\end{align}
where the last equation follows from Lemma \ref{NLEM::34}. Clearly, 
$$\frac{q-1}2+ \frac{q-1}{6}=\frac{2(q-1)}{3}\not\equiv 0\pmod{p}.$$ Hence, Equation (\ref{NEQ6::18}) is impossible, i.e., we get a contradiction.

Collecting the conclusions in Cases 1 and 2 yields $|S_1\cap S_2|\geq 3$ or $N\leq 1$. This completes the proof. 
\end{proof}

\begin{lemma}\label{NLEM9:11}
Let $q>5$ be an odd prime, $4\mid(q-1)$ and $3\mid(q+1)$. Then there exist two distinct integers $i$ and $j$ with $1\leq i, j\leq \frac{q-1}2$ such that $(\tr(\beta^i)+1)(\tr(\beta^j)+1)+1=0$.
\end{lemma}

\begin{proof}
 Similar to Lemma \ref{NLEM::35}, if $5\mid(q+1)$, we have 
\begin{align*}
(\tr(\beta^{\frac{q+1}5})+1)(\tr(\beta^{\frac{2(q+1)}{5}})+1)+1=0.
\end{align*}
The desired result follows. Below we suppose $5 \nmid(q+1)$. 

Let 
$$\Gamma=\left \{i: \ 1\leq i\leq \frac{q-1}2, \ i\neq \frac{q+1}3 \right\},$$ then $|\Gamma|=\frac{q-3}2$. By Result 3 of Lemma \ref{nlem30::33}, $\tr(\beta^i)\neq -1$ for any $i\in \Gamma$. Let 
$$T_1=\left\{\tr(\beta^i)+1: \ i\in \Gamma \right\}$$ and
 $T_2=\left\{ -(\tr(\beta^i)+1)^{-1}: \ i\in \Gamma \right\}$. By Result 1 of Lemma \ref{nlem30::33}, we deduce that $|T_1|=|T_2|=|\Gamma|$. Similar to  Lemma \ref{NLEM::35}, we have 
\begin{equation}\label{NEQ3::19}
|T_1\cap T_2|=2|\Gamma|-|T_1\cup T_2|,
\end{equation}
and $T_1 \cup T_2\subseteq  \gf(q)\backslash \{-1, 0, 1\}$. Similar to Lemma \ref{NLEM::35}, to prove the  desired result, it suffices to prove that $|T_1 \cap T_2|\geq 3$ or $$T:=|\{i\in \Gamma: \ (\tr(\beta^i)+1)^2+1=0 \}|\leq 1.$$  

We first prove $|T_1\cap T_2|\geq 2$ by distinguishing the following two cases. 

{\it Case 1}: Suppose $|T_1\cap T_2|=0$, then $T_1 \cup T_2=\gf(q)\backslash \{-1,0,1\}$. By Lemma \ref{NLEM::34}, we obtain that
\begin{align*}
0&=\sum_{t\in T_1\cup T_2} t= \sum_{t \in T_1} t+\sum_{t \in T_2} t\\
&=\sum_{i=1, \atop i\neq \frac{q+1}3}^{\frac{q-1}2}(\tr(\beta^i)+1)-\sum_{i=1, \atop  i\neq \frac{q+1}3}^{\frac{q-1}2}(\tr(\beta^i)+1)^{-1}\\ 
&=\sum_{i=1}^{\frac{q-1}2}(\tr(\beta^i)+1)-\sum_{i=1, \atop i\neq \frac{q+1}3}^{\frac{q-1}2}(\tr(\beta^i)+1)^{-1}=-1,
\end{align*}
a contradiction. 

{\it Case 2}: Suppose $|T_1\cap T_2|=1$, then there is an element $a\in \gf(q)\backslash \{-1, 0, 1\}$ such that $T_1 \cap T_2=\{a\}$. It is easily verified that if $a\in T_1\cap T_2$, then $-a^{-1}\in T_1\cap T_2$. Hence, $a^2=-1$. It follows from Equation (\ref{NEQ3::19}) that $|T_1\cup T_2|=q-4$. Consequently, there is an element $b\in \gf(q)\backslash \{-1, 0, 1\} $ such that $T_1\cup T_2=\gf(q)\backslash \{-1, 0, 1, b\}$. By Lemma \ref{NLEM::34}, we obtain that
$$-b=\sum_{t \in T_1\cup T_2} t=\sum_{t\in T_1} t +\sum_{t\in T_2} t-\sum_{t \in T_1\cap T_2} t=-1-a,$$
i.e., $b=1+a$. We now prove that $-a\in T_1 \cup T_2$. Clearly, to prove the desired result, we only need to prove $1+a\neq -a$. Suppose $1+a=-a$, then $a=\frac{q-1}2$. Consequently, 
\begin{equation}\label{NEQ8:20}
	a^2=(\frac{q-1}2)^2\equiv -1\pmod{q}.
\end{equation}
Suppose $q=4k+1$, by Equation (\ref{NEQ8:20}), we get that $k\equiv 1\pmod{q}$. It follows that $5\equiv 0\pmod{q}$. Consequently, $q=5$, a contradiction. That is to say, $-a\in T_1 \cup T_2$. Notice that $-(-a)^{-1}=-a$, we deduce that $-a\in T_1\cap T_2$, which contradicts the fact that $T_1\cap T_2=\{ a \}$.

Collecting the conclusions in Cases 1 and 2 yields $|T_1 \cap T_2|\geq 2$.

Finally, we prove that $|T_1 \cap T_2|\geq 3$ or $T\leq 1$. Suppose on the contrary that $|T_1 \cap T_2|=2$ and $T=2$, then $T_1 \cap T_2=\{a, -a\}$, where $a \in \gf(q)^*$ and $a^2=-1$. It follows from Equation (\ref{NEQ3::19}) that $|T_1\cup T_2|=q-5$. Consequently, there exist two distinct elements $t_1$ and $t_2$ in $\gf(q)\backslash \{-1, 0, 1\}$ such that $T_1 \cup T_2=\gf(q)\backslash \{-1, 0, 1, t_1,t_2\}$. By Lemma \ref{NLEM::34}, we obtain that
$$-(t_1+t_2)=\sum_{t \in T_1\cup T_2} t=\sum_{t\in T_1} t +\sum_{t\in T_2} t-\sum_{t \in T_1\cap T_2} t=-1.$$ 
Therefore, 
\begin{equation}\label{NEQ8:21}
t_1+t_2=1.	
\end{equation}
It is easily verified that if $t \in T_1\cup T_2$, then $-t^{-1} \in T_1 \cup T_2$. It follows that $t_2=-(t_1)^{-1}$. By Equation (\ref{NEQ8:21}), we deduce that 
\begin{equation}\label{NEQ8:22}
	t_1^2-t_1-1=0.
\end{equation}
It is easily verified that $x^2-x-1$ is reducible over $\gf(q)$ if and only if $5$ is a quadratic residue modulo $q$. By \cite{IR1990}, $5$ is a quadratic residue modulo $q$ if and only if $q\equiv \pm 1\pmod{5}$. We seek a contradiction by considering the following two cases.

{\it Case 1}: $5\nmid(q-1)$. By assumption, $5\nmid(q+1)$. Consequently, $x^2-x-1$ is irreducible over $\gf(q)$. It follows that Equation (\ref{NEQ8:22}) is impossible.   

{\it Case 2}: $5\mid(q-1)$. It is easily verified that $3\notin \{-1,0,1, t_1,t_2\}$. Consequently, $3\in T_1 \cup T_2$. By Result 2 of Lemma \ref{nlem30::33}, $3\notin T_1$. Hence, $3\in T_2$. It follows that there is an integer $i$ with $1\leq i\leq \frac{q-1}2$ such that $\tr(\beta^i)+1=-3^{-1}$. It then follows that $3\beta^{2i}+4\beta^{i}+3=0$. Notice that $20\mid(q-1)$, we deduce that $3x^2+4x+3$ is reducible over $\gf(q)$. Hence, $\beta^i\in \gf(q)$, which contradicts the fact that $1\leq i\leq \frac{q-1}2$.

Collecting the conclusions in Cases 1 and 2 yields $|T_1\cap T_2|\geq 3$ or $T\leq 1$. This completes the proof. 
\end{proof}

The parameters of $\C_{(q, q+1, 4, 1)}$ and $\C(4)$ are documented in the following theorem.

\begin{theorem}\label{NTH2::36}
Let $q\geq 9$ be an odd prime power. Then $\C_{(q, q+1, 4, 1)}$ has parameters $[q+1,q-5, d]$ and $\C(4)$ has parameters $[q+1, 6, q-5]$, where 
\begin{enumerate}
\item if $4\mid(q+1)$, then $d=4$;
\item if $4\mid (q-1)$ and $3\nmid(q+1)$, then $d=5$;
\item if $4\mid(q-1)$ and $3\mid(q+1)$, then $5\leq d\leq 6$, and $d=5$ provided that $5\nmid q$.  
 \end{enumerate}
\end{theorem}

\begin{proof}
It is clear that $\C_{(q, q+1, 4, 1)}$ has length $q+1$ and dimension $q-5$, and $\C(4)$ has length $q+1$ and dimension $6$, respectively. Notice that $d(\C(4))\in \{q-5,q-4\}$, and $d(\C(4))=q-4$ if and only if $\C(4)$ is MDS. Therefore, we only need to study the minimum distance of $\C_{(q, q+1, 4, 1)}$. 

{\it Case I}: $4\mid(q+1)$. By Lemma \ref{nlem:10}, we obtain that $d(\C_{(q, q+1, 4, 1)})=4$. The desired result follows.

{\it Case II}: $4\mid(q-1)$ and $3\nmid(q+1)$. By Lemma \ref{NLEM::35}, there exist two distinct integers $i$ and $j$ with $1\leq i, j\leq \frac{q-1}2$ such that $(\tr(\beta^i)+1)(\tr(\beta^j)+1)+1=0$. Let $x_1=\beta^i$, $x_2=\beta^{-i}$, $x_3=\beta^j$, $x_4=\beta^{-j}$ and $x_5=1$, then $|\{x_1,x_2,x_3,x_4, x_5\}|=5$. It is easily verified that
\begin{align*}
\sigma_{5,2}&= x_1x_2+(x_1+x_2)(x_3+x_4)+x_3x_4+(x_1+x_2+x_3+x_4)x_5	\\
&=2+\tr(\beta^i) \tr(\beta^j)+\tr(\beta^i) +\tr(\beta^j) \\
&=(\tr(\beta^i)+1)(\tr(\beta^j)+1)+1=0.
\end{align*}
By Lemma \ref{NLEM2::33}, we deduce that $d(\C_{(q, q+1, 4, 1)})=5$.

{\it Case III}: $4\mid(q-1)$ and $3\mid(q+1)$. In this case, $6\mid(q+1)$. It is clear that $$\C_{(q, q+1, 6, 1)}\subseteq \C_{(q, q+1, 4, 1)}.$$ Consequently, $d(\C_{(q, q+1, 4, 1)} )\leq d(\C_{(q, q+1, 6, 1)})$. By Lemma \ref{nlem:10}, we deduce that 
$$d(\C_{(q, q+1, 4, 1)} )\in \{5,6 \}.$$
Below we prove that $d(\C_{(q, q+1, 4, 1)} )=5$ for $5\nmid q$. 

Suppose $q=p^m$, where $p>5$ is an odd prime and $m$ is a positive integer. Since $3\mid(q+1)$, we deduce that $m$ is odd. Let $\gamma=\beta^{\frac{q+1}{p+1}}$, then $\gamma$ is a primitive $(p+1)$-th root of unity in $\gf(p^2)$. Let $\C$ be the cyclic code of length $p+1$ over $\gf(p)$ with generator polynomial $\m_{\gamma}(x)\m_{\gamma^2}(x)\m_{\gamma^3}(x)$, then $\C$ is the narrow-sense BCH code of length $p+1$ over $\gf(p)$ with designed distance $4$. Since $4\mid(q-1)$ and $m$ is odd, we get that $4\mid(p-1)$. Similar to Case II, we can prove that $d(\C)=5$. Consequently, there are five pairwise distinct integers $i_1,i_2,i_3,i_4, i_5$ with $0\leq i_j\leq p$ and $c_i\in \gf(p)^*$ such that $$c_1x^{i_1}+c_2x^{i_2}+c_3x^{i_3}+c_4x^{i_4}+c_5x^{i_5}\in \C.$$ 
It is easily verified that
$$c_1x^{{\frac{q+1}{p+1}}i_1}+c_2x^{\frac{q+1}{p+1}i_2}+c_3x^{\frac{q+1}{p+1}i_3}+c_4x^{\frac{q+1}{p+1}i_4}+c_5x^{\frac{q+1}{p+1}i_5}\in \C_{(q, q+1, 4,1)}.$$
 It then follows that $d(\C_{(q, q+1, 4, 1)})=5$. 
 
Collecting all the conclusions in Cases I, II and III, we complete the proof of this theorem.
\end{proof}

The parameters of $\overline{\C_4}$ and $(\overline{\C_4})^\perp$ are documented in the following theorem.

\begin{theorem}
Let $q\geq 9$ be an odd prime power. Then $\overline{\C_4}$ has parameters $[q+2,7, q-5]$ and $(\overline{\C}_4)^\perp$ has parameters $[q+2, q-5, d]$, where 
\begin{enumerate}
\item if $4\mid(q+1)$, then $d=5$;
\item if $4\mid (q-1)$ and $3\nmid(q+1)$, then $d=6$;
\item if $4\mid(q-1)$ and $3\mid(q+1)$, then $6\leq d\leq 7$, and $d=6$ provided that $5\nmid q$. 
 \end{enumerate}	
\end{theorem}

\begin{proof}
By Theorem \ref{thm-sdjoin1}, $\C_{4}$ has parameters $[q+1, 7, q-5]$. Then $\overline{\C_4}$ has length $q+2$ and dimension $7$. Consequently, $(\overline{\C_4})^\perp$ has length $q+2$ and dimension $q-5$. By Theorems \ref{NTHM29:1} and \ref{NTH2::36}, we obtain that $d(\overline{\C_4})=d(\C(4))=q-5$. By Theorems \ref{thm-fund21jproj} and \ref{NTH2::36}, we deduce that
$$d((\overline{\C_4})^\perp)=d(\C_{(q, q+1,4,1)})+1.$$ The desired conclusions on the minimum distance of $(\overline{\C_4})^\perp$ follow from Theorem \ref{NTH2::36}. This completes the proof. 
\end{proof} 

\begin{example}
Let $q=5^2$. Then the following hold.
\begin{enumerate}
	\item The code $\C_4$ has parameters $[26,7,20]$ and $(\C_4)^\perp$ has parameters $[26,19,8]$.
	\item The BCH code $\C_{(q, q+1, 4, 1)}$ has parameters $[26, 20, 5]$ and $\C(4)$ has parameters $[26,6,20]$.
	\item The code $\overline{\C_4}$ has parameters $[27,7,20]$ and $(\overline{\C_4})^\perp$ has parameters $[27,20,6]$.
\end{enumerate}	
\end{example}

For the case $u=4$, there is only one case where the minimum distance of $(\overline{\C_u})^\perp$ remains open, but we have the following conjecture.

\begin{conj}
	Let $q=5^m$ with $m\geq 3$ being odd, then $d(\C_{(q, q+1, 4, 1)})=5$ and $d((\overline{\C_4})^\perp)=6$.
\end{conj}

\subsection{The case $u=\lfloor \frac{q-1}2\rfloor$}  

Let $q>5$ be a prime power and $u=\lfloor \frac{q-1}2\rfloor$. Our task in this subsection is to settle the parameters of the code $\overline{\C_{u}}$ and its dual $(\overline{\C_{u}})^\perp$. First of all, we have the following result. 

\begin{corollary}
Let $q>5$ be an odd prime power and 	$u=\frac{q-1}2$. Then $\overline{\C_u}$ has parameters $[q+2, q-2, 4]$ and $(\overline{\C_u})^\perp$ has parameters $[q+2, 4, \frac{q+3}2]$. 
\end{corollary}

\begin{proof}
Since $q>5$, $u>2$. Note that $q+1\equiv 2\pmod{u}$, we deduce that $u\nmid(q+1)$. Clearly, $(u+1) \mid(q+1)$. The desired results follow from  Theorem \ref{nthm:11}.
\end{proof}

Starting from now on, let $q=2^m$ with $m\geq 3$ and $u=2^{m-1}-1$, then $\C_{(q, q+1, u, 1)}$ has checked polynomial $(x-1)\m_{\beta^{u}}(x)\m_{\beta^{u+1}}(x)$, where $\beta$ is a primitive $(q+1)$-th root of unity in $\gf(q^2)$. By Delsarte's theorem, we have 
\begin{equation}\label{n24:31}
\C_{(q, q+1, u, 1)}=\left\{ \bc_{(a, b, c)}  : a \in \gf(q), \ b, \ c \in \gf(q^2) \right\},	
\end{equation}
where $\bc(a, b, c)=(a+\tr(b\beta^{u i}+c\beta^{(u+1)i}))_{i=0}^q$, and $\tr(x)$ denote the trace function from $\gf(q^2)$ to $\gf(q)$. Let $\C(u)$ denote the dual of the BCH code $\C_{(q, q+1, u, 1)}$, then $\C(u)$  is the cyclic code of length $q+1$ over $\gf(q)$ with generator polynomial $(x-1)\m_{\beta^{u}}(x)\m_{\beta^{u+1}}(x)$. To determine the parameters of $\C_{(q, q+1, u, 1)}$ and $\C(u)$, we first prove the following lemma. 

\begin{lemma}\label{nlem27:36}
Let $q=2^m$ with $m\geq 3$. Then there are three pairwise distinct integers $i, j, k$ with $1\leq i,j,k\leq 2^{m-1}$ such that $\tr(\beta^i+\beta^j+\beta^k)=0$. 	
\end{lemma}

\begin{proof}
It is easily verified that  
$$x^{q+1}+1=(x+1)\prod_{i=1}^{2^{m-1}}(x^2+\tr(\beta^i) x+1),$$
and $\gcd(q+1,q)=1$. Consequently,  
\begin{equation}\label{NEQ24:32}
\tr(\beta^i)\neq \tr(\beta^j)	
\end{equation}
for any two distinct integers $i$ and $j$ with $1\leq i, j\leq 2^{m-1}$. We first prove that $\tr(\beta^i)\neq 0$ for each $1\leq i\leq 2^{m-1}$. Suppose $\tr(\beta^i)=0$, then $x^2+1$ divides $x^{q+1}+1$, which contradicts the fact that $\gcd(2,q+1)=1$. We now prove that $\tr(\beta)\neq 1$. Suppose $\tr(\beta)=1$, then $\beta^2+\beta+1=0$. It follows that $\beta^3=1$, which contradicts the fact that $\beta$ is a primitive $(q+1)$-th root of unity. Let 
$$S_1=\{\Tr(\beta^i): 1\leq i\leq 2^{m-1}\}$$ 
and $S_2=\{ \tr(\beta) \tr(\beta^i): 1\leq i\leq 2^{m-1}\}$, then $|S_1|=|S_2|=2^{m-1}$. It is easily verified that 
\begin{align}\label{NEQ27:23}
\tr(\beta) \tr(\beta^i)&=(\beta+\beta^{-1})(\beta^i+\beta^{-i}) \notag \\
	&=\beta^{i+1}+\beta^{-(i+1)}+\beta^{i-1}+\beta^{-(i-1)} \notag \\
	&=\tr(\beta^{i+1}+\beta^{i-1}).
\end{align}
Notice that $\tr(\beta^{2^{m-1}+1})=\tr(\beta^{2^{m-1}})$, we have 
\begin{equation}\label{NEQ27:24}
\tr(\beta) \tr(\beta^{2^{m-1}})=\tr(\beta^{2^{m-1}}+\beta^{2^{m-1}-1}).	
\end{equation} 
It is clear that $\tr(\beta) \tr(\beta)=\tr(\beta^2)$, then $|S_1\cap S_2|\geq 1$. If $|S_1\cap S_2|=1$, then
$$|S_1\cup S_2|=|S_1|+|S_2|-|S_1\cap S_2|=q-1.$$
Notice that $S_1\cup S_2\subseteq \gf(q)^*$, we get that
\begin{align*}
S_1\cup S_2=\{\tr(\beta^i): 1\leq i\leq 2^{m-1}\}\cup \{\tr(\beta) \tr(\beta^i): 2\leq i\leq 2^{m-1}\}=\gf(q)^*.	
\end{align*}
It then follows that
\begin{align*}
0=\sum_{x\in \gf(q)^*}x=\sum_{i=1}^{2^{m-1}} \tr(\beta^i)+\sum_{i=2}^{2^{m-1}} \tr(\beta) \tr(\beta^i)=\tr(\beta)^2,
\end{align*}
a contradiction. Therefore, $|S_1\cap S_2|\geq 2$. Then there is an integer $i$ with $2\leq i\leq 2^{m-1}$ such that 
$$\tr(\beta) \tr(\beta^i)=\tr(\beta^j),$$
 where $1\leq j\leq 2^{m-1}$. If $2\leq i\leq 2^{m-1}-1$, by Equation (\ref{NEQ27:23}), we deduce that
 $$\tr(\beta^{i+1}+\beta^{i-1})=\tr(\beta^j).$$ 
 Since $\tr(\beta^l)\neq 0$ for each $1\leq l\leq 2^{m-1}$, we deduce that $j\notin \{i-1,i+1\}$. The desired result follows. If $i=2^{m-1}$, by Equation (\ref{NEQ27:24}), we obtain that
 $$\tr(\beta^{2^{m-1}}+\beta^{2^{m-1}-1})=\tr(\beta^j).$$ 
 Similarly, we have $j\notin \{2^{m-1},2^{m-1}-1\}$. This completes the proof. 
\end{proof}

\begin{theorem}\label{thm-feb61} 
Let $q=2^m$ with $m \geq 3$ and $u=2^{m-1}-1$. Then $\C_{(q, q+1, u, 1)}$ has parameters $[q+1, 5, q-5]$ and $\C(u)$ has parameters $[q+1, q-4, 5]$.
\end{theorem}

\begin{proof}
It is clear that $\C_{(q, q+1, u, 1)}$ has length $q+1$ and dimension $5$. We now study the minimum distance of $\C_{(q, q+1, u, 1)}$. Let $\bc$ be a nonzero codeword of $\C_{(q, q+1, u, 1)}$. By Equation (\ref{n24:31}), there is $(a, b, c)\in \gf(q) \times \gf(q^2)^2$ and $(a, b, c)\neq (0,0,0)$ such that
 $$\bc=(a+\tr(b \beta^{ui}+c\beta^{(u+1)i}))_{i=0}^q.$$
 It is easily verified that $a+\tr(b \beta^{ui}+c\beta^{(u+1)i})=f_{a, b, c}(\beta^i)$, where  
 $$f_{a, b, c}(x)=a+b x^{2^{m-1}-1}+c x^{2^{m-1}}+c^q x^{2^{m-1}+1}+b^q x^{2^{m-1}+2}.$$ 
Note that 
$$f_{(a, b, c)}(x)^2 = a^2+b^2x^{q-2}+c^2 x^{q}+ c^{2q} x^{q+2}+b^{2q} x^{q+4}.$$
Consequently, 
\begin{align*}
f_{(a, b, c)}(\beta^i)^2&=a^2+b^2 \beta^{-3i}+c^2 \beta^{-i}+	 c^{2q} \beta^{i}+b^{2q} \beta^{3i}\\
&=(b^2+c^2\beta^{2i}+a^2 \beta^{3i}+c^{2q} \beta^{4i}+b^{2q} \beta^{6i}) \beta^{-3i}. 
\end{align*}
Let $g_{a, b, c}(x)=b^2+c^2 x^{2}+a^2 x^{3}+c^{2q} x^{4}+b^{2q} x^{6}$, then $f_{(a, b, c)}(\beta^i)=0$ if and only if $g_{a, b, c}(\beta^i)=0$. Since $(a, b, c)\neq (0,0,0)$, $g_{a, b, c}(x)\neq 0$. Consequently, $g_{a, b, c}(x)$ has at most $6$ solutions. Therefore, 
$$\wt(\bc) \geq q+1-6=q-5.$$ 
It follows that $d(\C_{(q, q+1, u,1)})\geq q-5$.

We now prove that $g_{a, b, c}(x)$ has $6$ solutions $x \in \{\beta^i:1\leq i\leq q \}$ for some $(a, b, c) \neq (0,0,0)$. By Lemma \ref{nlem27:36}, there are three pairwise distinct integers $i, j, k$ with $1\leq i,j,k\leq 2^{m-1}$ such that 
$$\tr(\beta^i+\beta^j+\beta^k)=0.$$ Clearly, $|\{\beta^{i}, \beta^{-i}, \beta^j, \beta^{-j}, \beta^k, \beta^{-k}\}|=6$.  Put $u=\tr(\beta^i)$,  $v=\tr(\beta^j)$, $w=\tr(\beta^k)$, then
\begin{align*}
\m_{\beta^i}(x)\m_{\beta^j}(x)\m_{\beta^k}(x)&=(x^2+u x+1)	(x^2+v x+1)(x^2+w x+1)\\
&=x^6+(uv+ uw+ v w+1) x^4+uvw x^3+(uv+uw+vw+1)x^2+1.		
\end{align*}
Set 
$$a=(uvw)^{q/2}, \ b=1, \ c=(uv+uw+vw+1)^{q/2}.$$ 
Then $\{\beta^{i}, \beta^{-i}, \beta^j, \beta^{-j}, \beta^k, \beta^{-k}\}$ is a set of $6$ solutions of $g_{a, b, c}(x)$. Consequently, the corresponding codeword $\bc$ has weight $q-5$. Therefore, $d(\C_{(q, q+1, u, 1)})=q-5$. 

It is clear that $\C(u)$ has length $q+1$ and dimension $q-4$. Note that $\C_{(q, q+1,u,1)}$ is not MDS, then $\C(u)$ is not MDS. As a result, $d(\C(u))\leq 5$. Notice that $\beta^i$ is a zero of $\C(u)$ for each 
$$i\in \{2^{m-1}-1, 2^{m-1}, 2^{m-1}+1, 2^{m-1}+2\}.$$ By the BCH bound, we deduce that $d(\C(u))\geq 5$. Therefore, $d(\C(u))=5$. This completes the proof.  
\end{proof} 

The parameters of the extended code $\overline{\C_u}$ are documented in the following theorem.

\begin{theorem}\label{thm-feb62} 
Let $q=2^m$ with $m \geq 3$, and $u=2^{m-1}-1$. Then $\overline{\C_u}$ has parameters $[q+2, q-3, 5]$ and $(\overline{\C_u})^\perp$ has parameters $[q+2, 5, q-4]$. 
\end{theorem}

\begin{proof}
By Theorem \ref{thm-sdjoin1}, $\C_{u}$ has parameters $[q+1, q-3, 5]$. Then $\overline{\C_{u}}$ has length $q+2$ and dimension $q-3$. Consequently, $(\overline{\C_{u}})^\perp$ has length $q+2$ and dimension $5$. By Theorems \ref{NTHM29:1} and \ref{thm-feb61}, we obtain that $d(\overline{\C_u})=d(\C(u))=5$. By Theorems \ref{thm-fund21jproj} and \ref{thm-feb61}, we deduce that
$$d((\overline{\C_u})^\perp)=d(\C_{(q, q+1,u,1)})+1=q-4.$$ 
This completes the proof. 
\end{proof} 

\begin{example}
Let $q=2^4$ and $u=2^{3}-1=7$. Then the following hold.
\begin{enumerate}
\item The code $\C_u$ has parameters $[17, 13, 5]$ and $(\C_u)^\perp$ has parameters $[17, 4, 14]$.	
\item The BCH code $\C_{(q, q+1, u, 1)}$ has parameters $[17,5,11]$ and $\C(u)$ has parameters $[17,12,5]$.
\item The extended code $\overline{\C_u}$ has parameters $[18, 13, 5]$ and $(\overline{\C_u})^\perp$ has parameters $[18,5,12]$.
\end{enumerate}
\end{example}

\subsection{The case $u=\lfloor\frac{q+1}2\rfloor$} 

Let $q>2$ be a prime power and $u=\lfloor\frac{q+1}2\rfloor$. Our task in this subsection is to settle the parameters of the code $\overline{\C_u}$ and its dual $(\overline{\C_u})^\perp$. By Theorem \ref{nthm:11}, we have the following result.

\begin{corollary}
Let $q$ be an odd prime power and $u=\frac{q+1}2$. Then $\overline{\C_u}$ has parameters $[q+2, q, 2]$ and $(\overline{\C_u})^\perp$ has parameters $[q+2,2,\frac{q+3}2]$.
\end{corollary}

Starting from now on, let $q=2^m$ with $m\geq 2$ and $u=2^{m-1}$. Recall that $\C_u$ is a $[q+1,q-1,3]$ MDS code over $\gf(q)$ and has generator polynomial $g_u(x)=\m_{\beta^{u}}(x)$, where $\beta$ is a primitive $(q+1)$-th root of unity in $\gf(q^2)$. Let $\C(u)$ be the dual of $\C_{(q, q+1, u, 1)}$, then $\C(u)$ is the cyclic code of length $q+1$ over $\gf(q)$ with generator polynomial $(x-1)\m_{\beta^u}(x)$.

\begin{lemma}\label{nlem:32}
Let $q=2^m$ with $m\geq 2$ and $u=2^{m-1}$. Then $\C_{(q, q+1, u, 1)}$ is a $[q+1, 3, q-1]$ MDS code over $\gf(q)$.	
\end{lemma}

\begin{proof}
To prove the desired result, we only need to prove that $\C(u)$ is an MDS code over $\gf(q)$. By definition, $\C(u)$ has length $q+1$, dimension $q-2$ and $d(\C(u))\leq 4$. Since $\gcd(u, q+1)=1$, we deduce that $\gamma=\beta^{u}$ is also a primitive $(q+1)$-th root of unity. Notice that $\gamma^{-1}=\beta^{-u}$, then $\gamma^{i}$ is a root of $(x-1)\m_{\beta^u}(x)$ for each $i\in \{-1, 0, 1\}$. By the BCH bound, $d(\C(u))\geq 4$. Therefore, $\C(u)$ is a $[q+1, q-2, 4]$ MDS code over $\gf(q)$. This completes the proof. 	
\end{proof}

\begin{theorem}\label{thm-sdjoint18}
Let $q=2^m$ with $m \geq 2$ and $u=2^{m-1}$. Then the following hold: 
\begin{enumerate}
\item $\C_{u}$ is a $[q+1, q-1, 3]$ MDS code over $\gf(q)$. 
\item $(\C_{u})^\perp$ is a $[q+1, 2, q]$ MDS code over $\gf(q)$. 
\item $\overline{\C_{u}}$ is a $[q+2, q-1, 4]$ MDS code over $\gf(q)$. 
\item $(\overline{\C_{u}})^\perp$ is a $[q+2, 3, q]$ MDS code over $\gf(q)$.    
\end{enumerate}
\end{theorem} 

\begin{proof}
The first two desired conclusions follow from Theorem \ref{thm-sdjoin1}. To  prove that the latter two desired conclusions, we only need to prove $d((\overline{\C_u})^\perp)=q$. From Theorem \ref{thm-fund21jproj} and Lemma \ref{nlem:32}, we deduce that $d((\overline{\C_u})^\perp)=d(\C_{(q, q+1, u, 1)})+1=q$. This completes the proof.  
\end{proof}

The codes $(\overline{\C_{u}})^\perp$ are the second family of hyperoval codes presented in this paper.  

\section{Summary and concluding remarks} \label{sec-final202171} 

It is observed that only a small number of families of NMDS codes with known weight distributions have been reported in \cite{DingTang19,HW22,LiHeng22a,LiHeng22b,TangDing20,WH2020,XCQ22,XuFanHan22}, as two NMDS codes over $\gf(q)$ with the same parameters $[n, k, n-k]$ may have different weight distributions. The first contribution of this paper is the several families of NMDS codes documented in Corollary \ref{Cor9::1}, Theorems \ref{NTHM:17}, \ref{NTHM18} and \ref{thm-sdjoint4}, and the settlement of the weight distributions of three families of NMDS codes (see Theorems \ref{thm-sdjoint3}, \ref{NTHM23} and \ref{thm-sdjoint4}). One of the infinite families of NMDS codes is both distance-optimal and dimension-optimal locally recoverable codes \cite{TFZD21}. The second contribution of this paper is the MDS codes $\overline{\C_2}$ documented in Corollary \ref{thm-sdjoint2} and the MDS codes $\C_{q/2}$ and $\overline{\C_{q/2}}$ documented in Theorem \ref{thm-sdjoint18} for even $q$. They are also both distance-optimal and dimension-optimal locally recoverable codes. The third contribution of this paper is the settlement of the parameters of $\overline{\C_{u}}$ and its dual for $u\in  \left\{2, 3,4, \lfloor\frac{q-1}2\rfloor, \lfloor\frac{q+1}2\rfloor\right\}$.

In this paper, we settled the parameters of the extended code $\overline{\C_u}$ for only a few values of $u$. It would be very challenging to study the parameters of the code $\overline{\C_u}$ and its dual $(\overline{\C_u})^\perp$ for $5 \leq u \leq \lfloor\frac{q-3}2\rfloor$. The reader is cordially invited to attack this problem. This shows that studying the extended code of a known linear code could be very challenging, even if the parameters of the linear code are known. This may explain why not much work on extended codes has been reported in the literature. 

Although the NMDS codes $\overline{\C_2}$, $\overline{\C_3}$ and $\overline{\C_4}$  have small dimensions (i.e., $3$, $5$ and $7$, respectively), they are interesting and valuable due to the following applications. These NMDS codes have applications in cryptography \cite{MS19}, finite geometry \cite{DeBoer96,FaldumWillems97,GP07} and distributed data storage systems (see the explanations in Section \ref{sec-intro}). Another application of these NMDS codes with small dimensions is the construction of subfield codes (i.e., the trace codes) with larger dimensions. We inform the reader that the binary subfield codes of $\overline{\C_2}$, $\overline{\C_3}$ and $\overline{\C_4}$ are distance-optimal or have the best parameters for $q=2^m$ with $m \in \{3,4,5,6\}$ in every case except one case (see Table \ref{table1}). The reader is referred to \cite{DH19,HD19,HDW20,WZ20} for information about subfield codes, which are different from the subfield subcodes of linear codes. 

\begin{table*}
{\caption{\rm The binary subfield codes of $\overline{\C_u}$
}\label{table1}
\begin{center}
\begin{tabular}{ccccc}\hline
$u$  & $m$ & subfield code  &  best distance & best dimension \\\hline
$2,3$    & $3$   &$[10,7,2]$ & optimal & $9$\\\hline
$2,3$    & $4$   & $[18,9,6]$   &  optimal & optimal \\\hline
$2,3$    & $5$   & $[34,11,12]$ & optimal & $12$ \\\hline
$2,3$    & $6$   & $[66,13,26]$ & best known & best known \\\hline
$4$    & $3$   & $[10,9,2]$     &  optimal & optimal \\\hline
$4$    & $4$   & $[18,17,2]$    &  optimal & optimal\\\hline
$4$    & $5$   & $[34,21,4]$    & $6$ & $27$  \\\hline
$4$    & $6$   & $[66,25,16]$    & best known & $29$ is best known\\\hline
\end{tabular}
\end{center}}
\end{table*}

According to our experimental data,  the two families of NMDS codes documented in Theorems \ref{thm-sdjoint3}, \ref{NTHM23} and \ref{thm-sdjoint4} do not support $2$-designs, only special NMDS codes support $2$-designs 
\cite{DingTang19,TangDing20,XCQ22}. 


\end{document}